\documentclass[conference]{IEEEtran}
\IEEEoverridecommandlockouts
\usepackage{cite}
\usepackage{amsmath,amssymb,amsfonts}
\usepackage{algorithmic}
\usepackage[pdftex]{graphicx}
\usepackage{subcaption}
\usepackage{multicol}
\usepackage{textcomp}
\usepackage{xcolor}
\usepackage{pgfplots}
\usepackage[utf8]{inputenc}
\usepackage[english]{babel}
\usepackage{amsthm}
\usepackage{csquotes}
\newtheorem{theorem}{Theorem}[section]
\newtheorem{lemma}[theorem]{Lemma}
\usepackage{url}
\usepackage{filecontents}
\usepackage{csvsimple}
\usepackage{times}
\usepackage{fancyhdr,amsmath,amssymb}
\usepackage[ruled,vlined,linesnumbered]{algorithm2e}

\usepackage[font=small,skip=0pt]{caption}
\pgfplotsset{width=10cm,compat=1.9}

\def\BibTeX{{\rm B\kern-.05em{\sc i\kern-.025em b}\kern-.08em
		T\kern-.1667em\lower.7ex\hbox{E}\kern-.125emX}}
\begin{document}
	
	\title{Dynamic Graph Configuration with Reinforcement Learning for Connected Autonomous Vehicle Trajectories}

	\author{
		\IEEEauthorblockN{Udesh Gunarathna, Hairuo Xie, Egemen Tanin,  Shanika Karunasekara, Renata Borovica-Gajic}
		\IEEEauthorblockA{pgunarathna@student.unimelb.edu.au,  \{xieh, etanin, karus, renata.borovica\}@unimelb.edu.au}
		\IEEEauthorblockA{\textit{School of Computing and Information Systems} \\
			\textit{University of Melbourne}}}
	
	\maketitle
	
	\begin{abstract}
		Traditional traffic optimization solutions assume that the graph structure of road networks is static, missing opportunities for further traffic flow optimization. We are interested in optimizing traffic flows as a new type of graph-based problem, where the graph structure of a road network can adapt to traffic conditions in real time. In particular, we focus on the dynamic configuration of traffic-lane directions, which can help balance the usage of traffic lanes in opposite directions. The rise of connected autonomous vehicles offers an opportunity to apply this type of dynamic traffic optimization at a large scale. The existing techniques for optimizing lane-directions are however not suitable for dynamic traffic environments due to their high computational complexity and the static nature.
		
		In this paper, we propose an efficient traffic optimization solution, called Coordinated Learning-based Lane Allocation (CLLA), which is suitable for dynamic configuration of lane-directions. CLLA consists of a two-layer multi-agent architecture, where the bottom-layer agents use a machine learning technique to find a suitable configuration of lane-directions around individual road intersections. The lane-direction changes proposed by the learning agents are then coordinated at a higher level to reduce the negative impact of the changes on other parts of the road network. Our experimental results show that CLLA can reduce the average travel time significantly in congested road networks. We believe our method is general enough to be applied to other types of networks as well.
		
	\end{abstract}
	
	
	\section{Introduction}
	The goal of traffic optimization is to improve traffic flows in road networks. Traditional solutions normally assume that the structure of road networks is static regardless of how traffic can change at real time~\cite{ford2009maximal,flowOverTime}. A less-common way to optimize traffic is by performing limited changes to road networks which when is-use are deployed in very small scale. We focus on dynamic lane-direction changes, which can help balance the usage of traffic lanes in many circumstances, e.g. as soon as when the traffic lanes in one direction become congested while the traffic lanes in the opposite direction are underused~\cite{Lane,reversiblelane}. Unfortunately, the existing techniques for optimizing lane-directions are not suitable for dynamic traffic environment at large scale due to their high computational complexity~\cite{Wu2009,dynamicLane_auto,dynamic_biLevel}. We develop an efficient solution for optimizing lane-directions in highly dynamic traffic environments. Our solution is based on an algorithm that modifies the property of a road network graph for improving traffic flow in the corresponding road network, introducing a new graph problem. 
	
	\begin{figure}[htbp]
		\centering
		\begin{subfigure}[b]{.16\textwidth}
			\includegraphics[width=\linewidth, height=3cm]{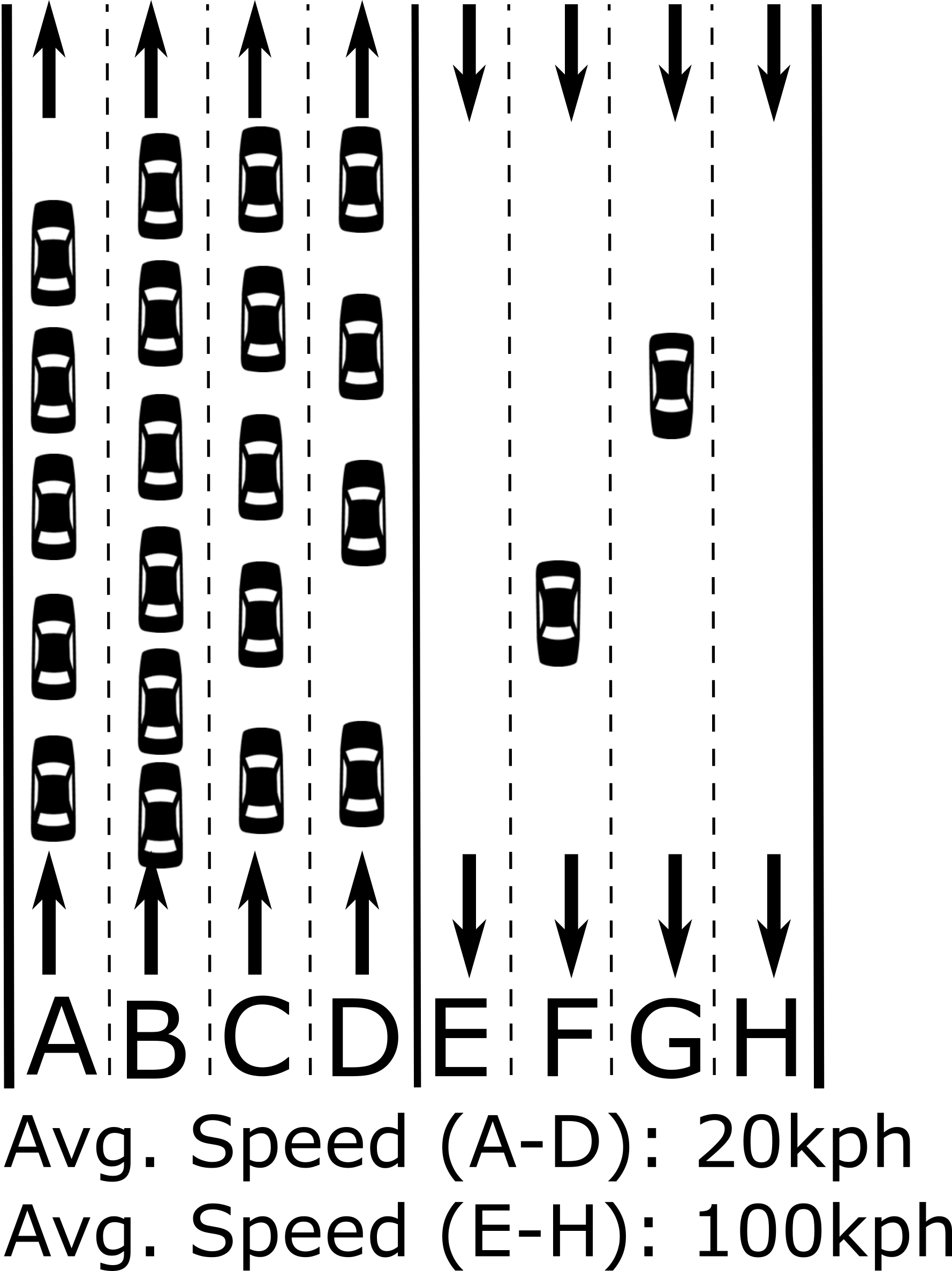}
			\caption{Traffic before lane-direction change}
			\label{fig-expBfChg}
		\end{subfigure}\hspace{0.05\textwidth}
		\begin{subfigure}[b]{.15\textwidth}
			\includegraphics[width=\linewidth, height=3cm]{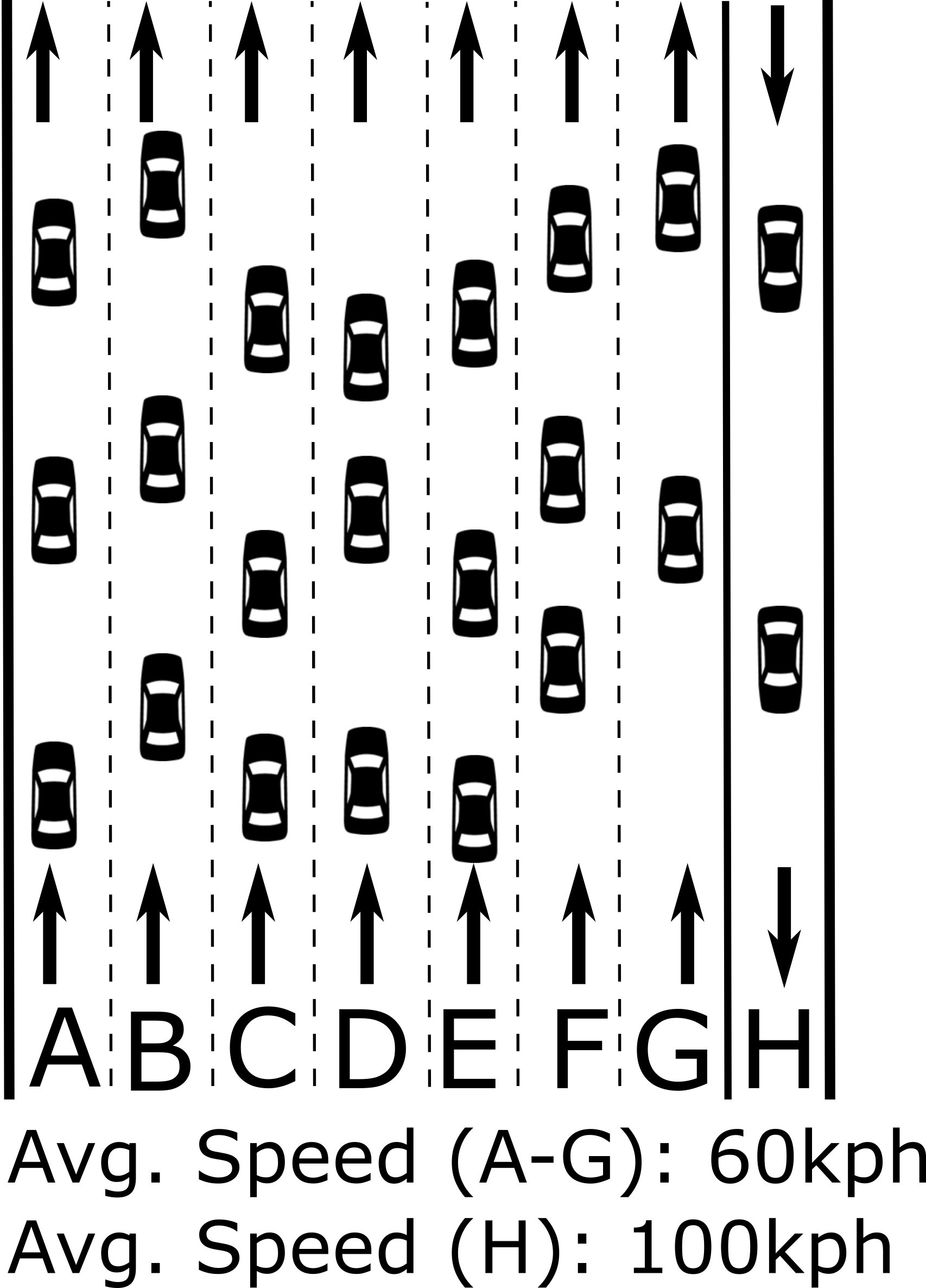}
			\caption{Traffic after lane-direction change}
			\label{fig-expAfChg}
		\end{subfigure}
		\caption{The impact of lane-direction change on traffic flow. There are 20 vehicles moving in the north-bound direction and 2 vehicles moving in the south-bound direction.}
		\label{fig-expDirChg}
		\vspace{-3mm}
	\end{figure}
	
	The impact of dynamic lane-direction configurations can be shown in the following example, where 20 vehicles are moving north-bound and 2 vehicles are moving south-bound (Figure~\ref{fig-expDirChg}) at a certain time. In Figure~\ref{fig-expBfChg}, there are 4 north-bound lanes and 4 south-bound lanes. Due to the large number of vehicles and the limited number of lanes, the north-bound traffic is highly congested. At the same time, the south-bound vehicles are moving at a high speed as the south-bound lanes are almost empty. Figure~\ref{fig-expAfChg} shows the dramatic change of traffic flow after lane-direction changes are applied when congestion is observed, where the direction of E, F and G is reversed. The north-bound vehicles are distributed into the additional lanes, resulting in a higher average speed of the vehicles. At the same time, the number of south-bound lanes is reduced to 1. Due to the low number of south-bound vehicles, the average speed of south-bound traffic is not affected. The lane-direction change helps improve the overall traffic efficiency in this case. This observation was used by traffic engineers of certain road segments for many years and applied in a more static way. We aim to scale this to extreme levels in time and space. The benefit of dynamic lane-direction changes can also be observed in preliminary tests, where we compare the average travel time of vehicles in two scenarios, one uses dynamic lane-direction configurations, another uses static lane-direction configurations. The dynamic lane-direction configurations are computed with a straightforward solution (Section~\ref{demand based}). The result shows that lane-direction changes reduce travel times by  14\% on average when the traffic increases (see Figure~\ref{lane change plot}). In traffic engineering terms this is a dramatic reduction.  
	
	\pgfplotstableread[col sep = comma]{demandbased.csv}\loadedtable
	
	\begin{figure}
		\begin{tikzpicture}
		\begin{axis}[
		height=5cm, width=7.5cm,
		legend entries = {Dynamic configurations, Static configurations},
		legend cell align={left},
		legend style={at={(0,1)},anchor=north west},
		ylabel={Average travel time (min)},
		xlabel={The number of vehicles generated},
		]
		\addplot table
		[x=a, y=c, col sep=comma] {\loadedtable};
		\addplot table
		[x=a, y=b, col sep=comma] {\loadedtable};
		
		\end{axis}
		\end{tikzpicture}
		\caption{The average travel time of vehicles when using static and dynamic lane-direction configurations.}
		\label{lane change plot}
		\vspace{-4mm}
	\end{figure}
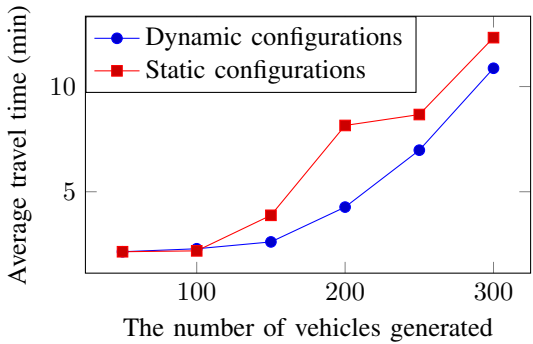

	Despite their potential benefit, dynamic lane-direction changes cannot be easily applied to existing traffic systems as they require additional signage and safety features~\cite{hntb10}. The emergence of connected autonomous vehicles (CAVs)~\cite{Narla2013} however can make dynamic lane-direction changes a common practice in the future. Our previous work shows that CAVs have the potential to enable innovative traffic management solutions~\cite{Ramamohanarao}. Compared to human-driven vehicles, CAVs are more capable of responding to a given command in a timely manner~\cite{dynamicLane_auto}. CAVs can also provide detailed traffic telemetry data to a central traffic management system in real time. This helps the system to adapt to dynamic traffic conditions. 
	
	We formulate lane allocation based on real-time traffic as a new graph problem with the aim to find a new graph $G^\prime_t$ from a road network ($G$) (i.e., dynamically optimize the graph) such that total travel time of all vehicles in the road network is minimized.
	In order to optimize the flow of the whole network, all the traffic lanes in the network must be considered. In many circumstances, one cannot simply allocate more traffic lanes at a road segment for a specific direction when there is more traffic demand in the direction. This is because a lane-direction change at a road segment can affect not only the flow in both directions at the road segment but also the flow at other road segments. Due to the complexity of the problem, the computation time can be very high with the existing approaches as they aim to find the optimal configurations based on linear programming~\cite{Wu2009,dynamicLane_auto,dynamic_biLevel}, and hence are not suitable for frequent recomputation over large networks.
	
	To address the above mentioned issues, we propose a light-weight and effective framework, called a Coordinated Learning-based Lane Allocation (CLLA) framework, for optimizing lane-directions in dynamic traffic environments. 
	The CLLA approach finds the configurations that effectively improve the traffic efficiency of the whole network, while keeping the computation cost of the solution low. The key idea is that traffic optimization can be decoupled into two processes: i) a local process that proposes lane-direction changes based on local traffic conditions around road intersections, and b) a global process that evaluates the proposed lane-direction changes based on their large scale impact. 
	
	\begin{figure}
		\centering
		\includegraphics[width=.5\linewidth, height=3cm]{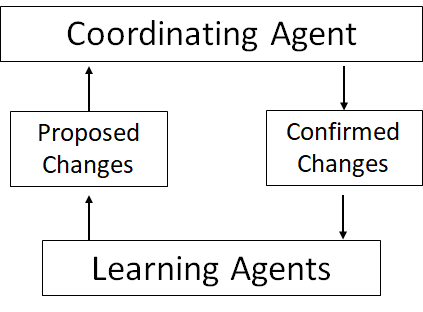}
		\caption{The hierarchical architecture of our traffic management solution based on lane-direction changes.}
		\label{intro_fig}
		\vspace{-4mm}
	\end{figure}
	
	The architecture of our hierarchical solution is illustrated in Figure~\ref{intro_fig}. The bottom layer consists of a set of autonomous agents that operate at the intersection level. An agent finds suitable lane-direction changes for the road segments that connect to a specific intersection. The agent uses reinforcement learning~\cite{sutton}, which helps determine the best changes based on multiple dynamic factors. The agents send the proposed lane-direction changes to the upper layer, which consists of a coordinating agent. The coordinating agent maintains a data structure, named \emph{Path Dependency Graph (PDG)}, which is built based on the trip information of connected autonomous vehicles. With the help of the data structure, the coordinating agent evaluates the global impact of the proposed lane-direction changes and decides what changes should be made to the traffic lanes. The decision is sent back to the bottom layer agents, which will make the changes accordingly. 
	
	The main contributions of our work are as follows:
	\begin{itemize}
		\item We formulate dynamic lane allocation as a new graph problem (Dynamic Resource Allocation problem).
		
		\item We propose a hierarchical multi-agent solution (called CLLA) for efficient dynamic optimization of lane-directions that uses reinforcement learning to capture dynamic changes in the traffic. 
		
		\item We introduce an algorithm and innovative data structure (called \emph{path dependency graph}) for coordinating lane-direction changes at the global level. 
		
		\item Extensive experimental evaluation shows that CLLA significantly outperforms other traffic management solutions, making it a viable tool for mitigating traffic congestion for future traffic networks.
	\end{itemize}
	
	\section{Related Work}
	\subsection{Traffic Optimization Algorithms} 
	\label{sec-relatedOptimizationAlgorithms}
	Existing traffic optimization algorithms are commonly based on traffic flow optimization with linear programming~\cite{Ford,flowOverTime,K2009}. The algorithms are suitable for the situations where traffic demand and congestion levels are relatively static. When there is a significant change in the network, the optimal solutions normally need to be re-computed from scratch. Due to the high computational complexity of finding an optimal solution, the algorithms may not be suitable for highly dynamic traffic environments. 
	
	With the rise of reinforcement learning~\cite{Ravishankar2017}, a new generation of traffic optimization algorithms have emerged~\cite{Walraven2016, Yau2017, Mannion}. In reinforcement learning, a learning agent can find the rules to achieve an objective by repeatedly interacting with an environment. The interactive process can be modelled as a finite Markov Decision Process, which requires a set of states $S$ and a set of actions $A$ per state. Given a state $s$ of the environment, the agent takes an action $a$. As the result of the action, the environment state may change to $s^\prime$ with a reward $r$. The agent then decides on the next action in order to maximize the reward in the next round. Reinforcement learning-based approaches can suggest the best actions for traffic optimization given a combination of network states, such as the queue size at intersections~\cite{Aslani2018, Arel2010}. They have an advantage over linear programming-based approaches, since if trained well, they can optimize traffic in a highly dynamic network. In other words, there is no need to re-train the agent when there is a change in the network. For example, Arel et al. show that a multi-agent system can optimize the timing of adaptive traffic lights based on reinforcement learning~\cite{Arel2010}. Different to the existing approaches, our solution uses reinforcement learning for optimizing lane-directions.
	
	A common problem with reinforcement learning is that the state space can grow exponentially when the dimensionality of the state space grows linearly. For example, let us assume that the initial state space only has one dimension, the queue size at intersections. If we add two dimensions to the state space, traffic signal phase and traffic lane configuration, there will be three dimensions and the state space is four times as large as the original state space. The fast growth of the state space can make reinforcement learning unsuitable for real deployments. This problem is known as the \emph{curse of dimensionality}~\cite{Multiagent_palnning}. A common way to mitigate the problem is by using a function approximator such as a neural network. Such techniques have been mainly used for dynamic traffic signal control~\cite{Genders, Pol2019}, while we extend the use of the technique to dynamic lane-direction configurations.
	
	Many existing traffic optimization solutions use model-based reinforcement learning, where one needs to know the exact probability that a specific state transits to another specific state as a result of a specific action~\cite{Wiering,Steingr2005reinforcementlearning}. Nonetheless, such an assumption is unrealistic since the full knowledge of state transition probabilities can hardly be known for highly complex traffic systems. Different to model-based approaches, our optimization solution employs a model-free algorithm, Q-learning~\cite{q-learning}, which does not require such  knowledge and hence is much more applicable to real traffic systems. 
	
	Coordination of multi-agent reinforcement learning can be achieved through a joint state space or through a coordination graph \cite{CoG}. Such techniques however require agents to be trained on the targeted network. Since our approach uses an implicit mechanism to coordinate, once an agent is trained, it can be used in any road network. 
	
	\subsection{Lane-direction Configurations}
	
	Research shows that dynamic lane-direction changes can be an effective way to improve traffic efficiency~\cite{Lane}. However, existing approaches for optimizing lane-directions are based on linear programming~\cite{dynamicLane_auto,Zhang2007, Wu2009,dynamic_biLevel}, which are unsuitable for dynamic traffic environments dues to their high computational complexity. For example, Chu et al. use linear programming to make lane-allocation plans by considering the schedule of connected autonomous vehicles~\cite{dynamicLane_auto}. Their experiments show that the total travel time can be reduced. However, the computational time grows exponentially when the number of vehicles grows linearly, which can make the approach unsuitable for highly dynamic traffic environments. Other approaches perform optimization based on two processes that interact with each other~\cite{Zhang2007, Wu2009,dynamic_biLevel}. One process is for minimizing the total system cost by reversing lane directions while the other process is for making route decisions for individual vehicles such that all the vehicles can minimize their travel times. To find a good optimization solution, the two processes need to interact with each other iteratively. The high computational cost of the approaches can make them unsuitable for dynamic traffic optimizations. Furthermore, all these approaches assume exact knowledge of traffic demand over the time horizon is known beforehand; this assumption does not hold when traffic demand is stochastic \cite{lane_cell_model}. On the contrary, CLLA is lightweight and can adapt to highly dynamic situations based on reinforcement learning. The learning agents can find the effective lane-direction changes for individual road intersections even when traffic demand changes dramatically.

	\subsection{Traffic Management with Connected Autonomous Vehicles}
	Some recent development of traffic management solutions is tailored for the era of connected autonomous vehicles. Telebpour and Mahmassani develop a traffic management model that combines connected autonomous vehicles and intelligent road infrastructures for improving traffic flow~\cite{CAV_tm}. Guler et al. develop an approach to improve traffic efficiency at intersection using connected autonomous vehicles~\cite{CAV_I}. We use the CAVs as an opportunity for lane optimization. To the best of our knowledge, we are the first to study dynamic lane-direction changes at large scale networks in the era of connected autonomous vehicles.
	
	\section{Problem Definition} \label{prob def}
	
	In this section, we formalize the problem of traffic optimization based on dynamic configuration of lane directions. Our problem is similar to Network Design Problem \cite{NDP}, however NDP is based on the assumption of knowledge of traffic demand for whole time horizon at time zero and the output network is designed for a common state. We try to configure a graph (road network) at regular time intervals based on real-time traffic, thus we name this problem,  \emph{Dynamic Graph Resource Allocation} problem. 
	
	Let $G(V,E)$ be a road network graph, where $V$ is a set of vertices and $E$ is a set of edges. Let us assume that edge $e \in E$ connects between vertex $x\in V$ and vertex $y\in V$. The edge has three properties. The first property is the the total number of lanes, $n_{e}$, which is a constant number. The second property is the number of lanes that start from $x$ and end in $y$, $n_{e_1}$. The third property is the number of lanes in the opposite direction (from $y$ to $x$), $n_{e_2}$. $n_{e_1}$ and $n_{e_2}$ can change but they are always subject to the following constraint.
	\begin{equation}
	n_{e_1} + n_{e_2} = n_{e}
	\end{equation}
	
	We assume that a CAV follows a pre-determined path based on an origin-destination (O-D) pair. Let the number of unique O-D pairs of the existing vehicles be $k$ at a given time $t$. For the $i^{th}$ ($i<=k$) O-D pair, let $d_{i,t}$ be the traffic demand at time $t$, i.e., the number of vehicles with the same O-D pair at that time. The traffic demand can be stochastic. Let the travel time of vehicle $j$ with the $i^{th}$ O-D pair be $TT_{i,j}$, which is the duration for the vehicle to move from the origin to destination.
	
	For a given time $t$, the average travel time of all the vehicles, which will reach their destination during a time period $T$ after $t$, can be defined as 	
	\begin{equation}
	ATT_{<t,t+T>}=\sum_{i=1}^{k} \sum_{j=1}^{m_{i}} TT_{i,j} / \sum_{i=1}^{k}m_{i}
	\end{equation}
	where $m_{i}$ is the number of vehicles with the $i^{th}$ O-D pair that will complete their trips between $t$ and $t+T$.
	
	We define and solve a version of the problem where at frequent regular intervals we optimize travel time, while changing the lane arrangement in all edges. 
		We find a new graph $G^\prime_{t}(V,E^\prime)$ at a given time $t$ from previous $G$ at previous time step. Let $e^\prime_1, e^\prime_2 \in E^\prime$ and $e^\prime_1$ connects from vertex $x$ to vertex $y$ and $e^\prime_2$ connects vertex $y$ to vertex $x$. We find for all edges the values of $n_{e^\prime_1}$ and $n_{e^\prime_2}$, such that the average travel time  $ATT_{<t,t+T>}$ is minimized. We call this Dynamic Resource Allocation problem.
	
	\section{Demand-based Lane Allocation (DLA)} \label{demand based}
	When considering dynamic lane-direction changes, a straightforward solution can use a centralized approach to optimize lane-directions based on the full knowledge of traffic demand, i.e., the number of vehicle paths that pass through the road links. We call this solution \emph{Demand-based Lane Allocation (DLA)}. Algorithm~\ref{baseline_algo} shows the implementation (in pseudo code) of such an idea to compute the configuration of lane-directions. DLA allocates more lanes for a specific direction when the average traffic demand per lane in the direction is higher than the average traffic demand per lane in the opposite direction. To specify the directions, we define two terms, \emph{upstream} and \emph{downstream}. The terms are defined as follows. Let us assume that all the vertices of the road network graph are ordered by the identification number of the vertices. Given two vertices, $v_1$ and $v_2$, and a direction that points from $v_1$ to $v_2$, we say that the direction is \textbf{upstream} if $v_1$ is lower than $v_2$ or \textbf{downstream} if $v_1$ is higher than $v_2$. 
	
	DLA first computes the traffic demand at the edges that are on the path of the vehicles (Line 1-6). The traffic demand is computed for the upstream direction ($up_e$) and the downstream direction ($down_e$) separately. Then it quantifies the difference of the average traffic demand per lane between the two directions (Line 9-11). Based on the difference between the two directions, DLA decides whether the number of lanes in a specific direction needs to be increased (line number 11-14). We should note that increasing the number of lanes in one direction implies that the number of lanes in the opposite direction is reduced. DLA only reduces the number of lanes in a direction if the traffic demand in that direction is lower than a threshold (Line 12). The complexity of the algorithm is $\mathcal{O}(k|E|)$, where $|E|$ is the number of edges in $G$ and $k$ is the number of O-D pairs.
	
	While straightforward to implement and effective, there are two notable drawbacks of DLA. First, the algorithm does not consider real-time traffic conditions, such as the queue length at a given time, during optimization; the only information used for optimization is (assumed apriori known) traffic demand and exact knowledge of traffic demand is difficult to obtain in dynamic road networks \cite{lane_cell_model}. This can make the lane-direction configuration less adaptive (and less applicable) to real-time traffic conditions. Second, the lane-direction optimization for individual road segments is performed individually, not considering the potential impact of a lane-direction change at a road segment on other road segments in the same road network. Therefore, a lane-direction change that helps improve traffic efficiency at a road link may lead to the decrease of traffic efficiency in other parts of the road network.
	\setlength{\textfloatsep}{0pt}
	\begin{algorithm} 	
		\KwIn{A road network graph $G(V,E)$.}
		\KwIn{The set of paths. A path is a sequence of edges on the shortest path between a specific Origin-Destination (O-D) pair. The set of paths includes the paths of all unique O-D pairs.}
		\KwIn{The demands associated with the paths, where a demand is the number of vehicles that follow a specific path.}
		\KwIn{$th$: demand threshold.} 
		\KwIn{$g$: minimal gap in traffic demand that can trigger a lane-direction change.}
		
		\ForEach{$p \in$ paths}{
			\ForEach{$e \in p$}{
				\If{$p$ passes $e$ in the upstream direction}{
					$up_e$ += demand of $p$
				}
				\If{$p$ passes $e$ in the downstream direction}{
					$down_e$ += demand of $p$
				}	
			}
		}
		
		\ForEach{$e \in E$}{
			
			$minLoad$ $\gets$ min($up_e, down_e$)
			
			$down_e' \gets down_e$ / number of downstream lanes
			
			$up_e' \gets up_e$ / number of upstream lanes
			
			$gap$ $\gets$ $\dfrac{down_e' - up_e'}{up_e' + down_e'}$
			
			\If{$minLoad < th$}{
				\If{$gap >  g$}{
					move one upstream lane to the set of downstream lanes 
				}
				\If{$gap < -g$}{
					move one downstream lane to the set of upstream lanes
				}
			}
			
		}
		\caption{\bf Demand-based Lane Allocation (DLA)}
		\label{baseline_algo}
		
	\end{algorithm}

\section{Coordinated Learning-based Lane Allocation (CLLA)} 

To tackle the problems of the straightforward solution, we propose a fundamentally different solution, a Coordinated Learning-based Lane Allocation (CLLA) framework. CLLA uses a machine learning technique to help optimize lane-direction configurations, which allows the framework to adapt to a high variety of real-time traffic conditions. In addition, CLLA coordinates the lane-direction changes by considering the impact of a potential lane-direction change on different parts of the road network. DLA, on the other hand, does not consider the global impact of lane-direction changes. Another difference between the two is that DLA requires the full path of vehicles to be known for computing traffic demand. As detailed later, CLLA only needs to know partial information about vehicle paths in addition to certain information about real-time traffic conditions, such as intersection queue lengths and lane configuration road segments which can be obtained from inductive-loop traffic detectors. 

CLLA uses a two-layer multi-agent architecture. The bottom layer consists of learning agents that are responsible for optimizing the direction of lanes connected to specific intersections. Using the multi-agent approach can significantly boost the speed of learning. The lane-direction changes that are decided by the learning agents are aggregated and evaluated by a coordinating agent at the upper layer, which will send the globally optimized lane-direction configuration to the bottom layer for making the changes. 

\begin{figure}[htbp]
    \centering
	\includegraphics[width=\linewidth]{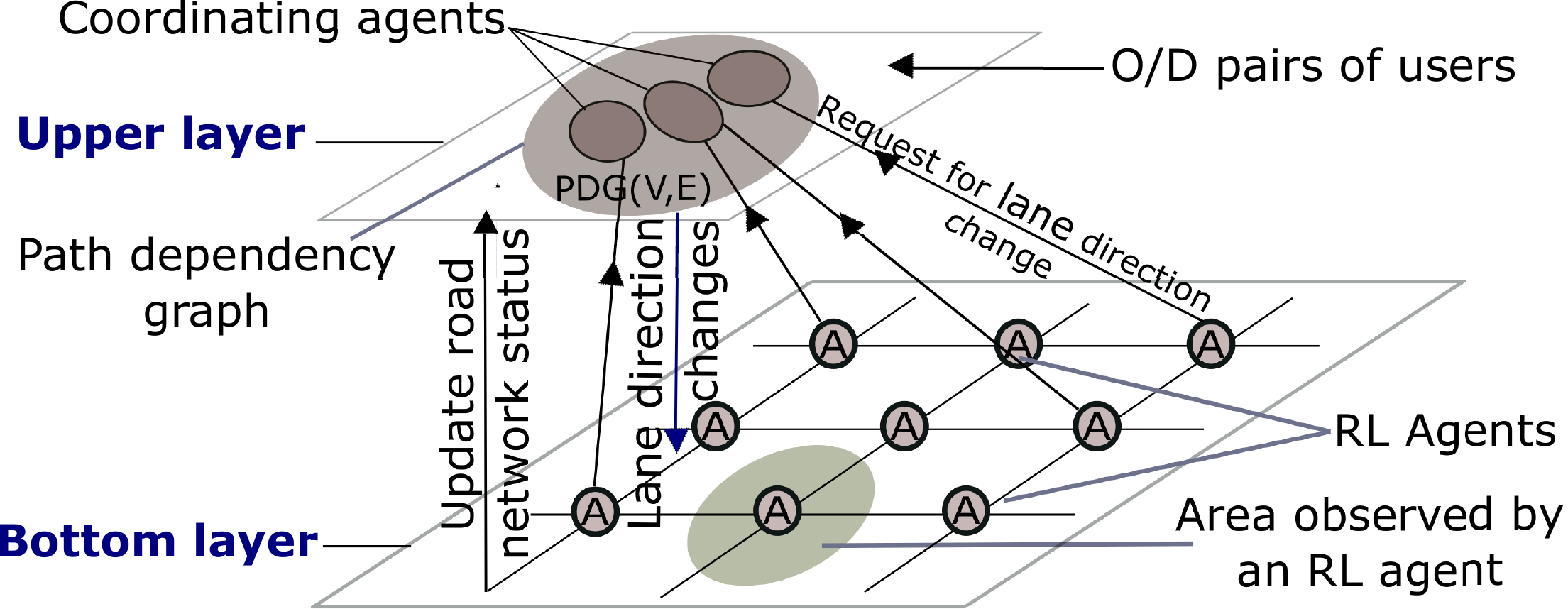}
	\caption{An overview of the CLLA's architecture }
	\label{overall}
\end{figure}
A more detailed overview of CLLA is shown in Figure~\ref{overall}. As the figure shows, an agent in the bottom layer only observes the local traffic condition around a specific intersection. Agents make decisions on lane-direction changes independently. Whenever an agent needs to make a lane-direction change, it sends the proposed change to the coordinating agent in the upper layer. The agents also send certain traffic information to the upper layer periodically. The information can help indicate whether there is an imbalance between upstream traffic and downstream traffic at specific road segments. The coordinating agent evaluates whether a change would be beneficial at the global level. The evaluation process involves a novel data structure, \emph{Path Dependency Graph (PDG)}, to inform decisions sent from the bottom layer. The coordinator may allow or deny a lane-direction change request from the bottom-layer. It may also decide to make further changes to lane-directions in addition to the requested changes. After evaluation, the coordinating agent informs the bottom-layer agents of the changes to be made.

\begin{figure}[htbp]
    \centering
	\includegraphics[width=.85\linewidth]{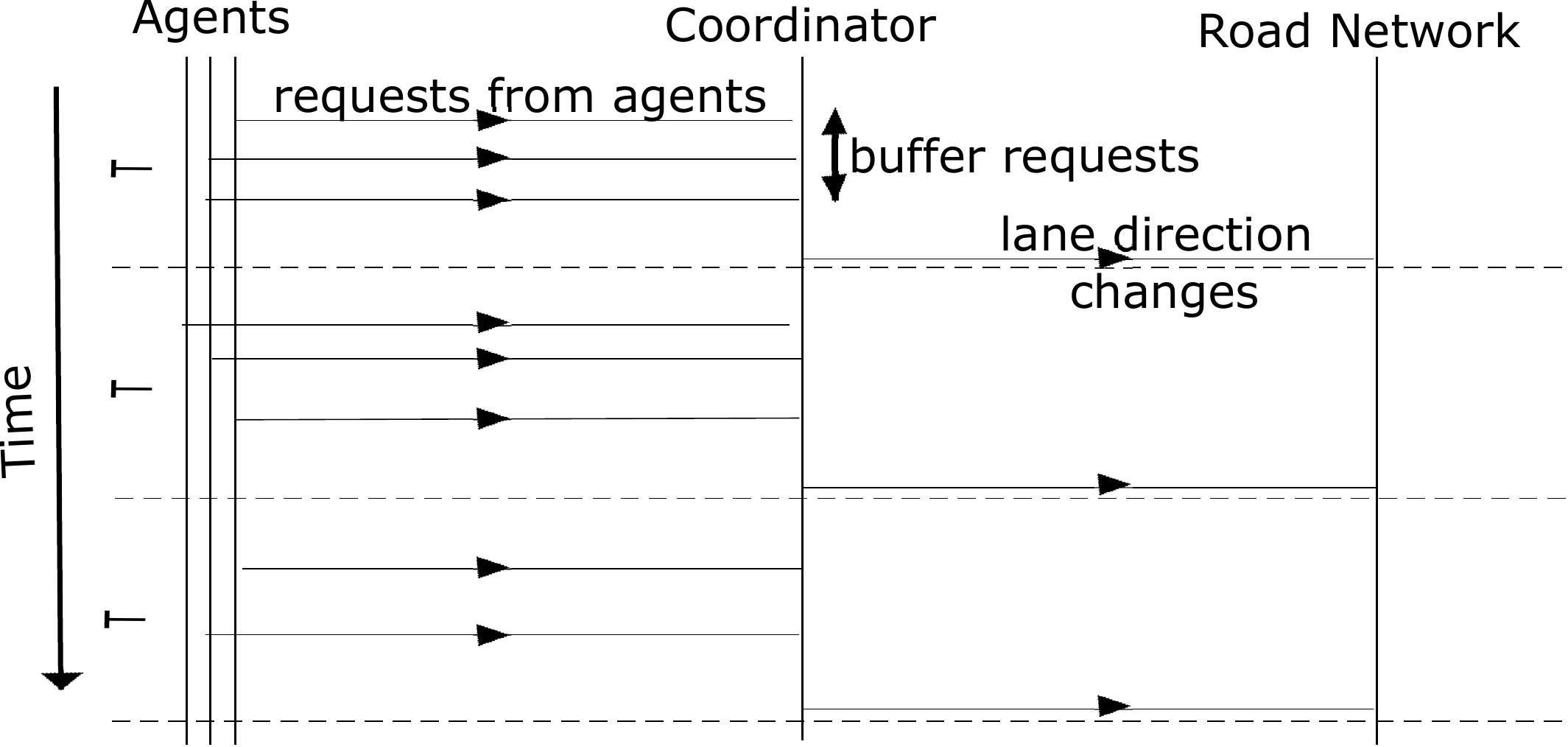}
	\caption{The CLLA's communication timeline between agents and the coordinator}
	\label{time pdg}
\end{figure}
We should note that the coordinator does not need to evaluate a lane-direction change request as soon as it arrives. As shown in Figure~\ref{time pdg}, the coordinator evaluates the lane-direction changes periodically. The time interval between the evaluations is $T$. All the requests from the bottom-layer agents are buffered during the interval. The exact value of the interval needs to be adjusted case by case. A short interval may increase the computational cost of the solution. A long interval may decrease the effectiveness of the optimization.

\subsection{CLLA Algorithm}
\label{sec-CLLA}
Algorithm~\ref{alg-CLLA} shows the entire optimization process of CLLA. During one iteration of the algorithm, each learning agent finds the lane-direction changes around a specific road intersection using the process detailed in Section~\ref{sec-qLearning}. The proposed change is stored as an edge-change pair, which is buffered in the system (Line 5). When it is time to evaluate the proposed changes, the system uses the \emph{Direction-Change Evaluation} algorithm (Section~\ref{sec-coord}) to quantify the conflicts between the proposed changes (Line 8). For example, when a learning agent proposes to increase the number of upstream lanes on road segment $s_1$ while another agent proposes a lane-direction change on a different road segment $s_2$, which can lead to the increase of the downstream traffic flow on $s_1$, there is a conflict between the proposed changes for $s_1$. The Change Evaluation algorithm also expands the set of the proposed changes that may be beneficial. Upon returning from the Change Evaluation algorithm, CLLA checks the expanded set of edge-change pairs. For each edge-change pair, if the number of conflicts for the edge are below a given limit, the change is applied to the edge (Line 12).

\begin{algorithm}
   
    \KwIn{$T$, action evaluation interval}

	\KwIn{$EC_{initial}$, set of edge-change pairs proposed by the learning agents}
	
	\KwIn{$EC_{expanded}$, set of edge-change pairs given by the coordinator}
	
   	\KwIn{$mc$, the maximum number of conflicts between lane-direction changes before a proposed lane-direction change can be applied to an edge.}

	\KwIn{$G(V,E)$, a road network graph.}
	
	\KwIn{$l$, the lookup distance for building Path Dependency Graph.}
	
	\KwIn{$dp$, the maximum search depth in Path Dependency Graph for evaluating lane-direction changes.}

    \While{True}{
        \ForEach{$agent \in Agents $}{
        	determine the best lane-direction change for all the edges (road segments) that connect to the vertex (road intersection) controlled by the $agent$ 
        	
             \ForEach{edge $e$ that needs a lane-direction change}{
                $EC_{initial}.insert(\{e,change\})$
                }
            }
        
        \If{T = t}{ 
            $t \gets 0$
            
           $EC_{expanded}$ $\gets$ \textbf{Direction-Change Evaluation ($EC_{initial}$, $G$, $l$, $dp$)}
            
			$EC_{initial} \gets \emptyset$ 
            
            \ForEach{$\{e,change\}$ in $EC_{expanded}$}{
                \If {number of conflicts on $e$ is less than $mc$}
                    {apply the lane-direction change to $e$}
                }
            }
                   $t \gets t + 1$
        }
 
        \caption{\bf Coordinated Lane Allocation (CLLA)}
        \label{alg-CLLA}

\end{algorithm}

\subsection{Learning-based Lane-direction Configuration}
\label{sec-qLearning}
In the CLLA framework, the bottom-layer agents use the Q-learning technique to find suitable lane-direction changes based on real-time traffic conditions. Q-learning aims to find a policy that maps a state to an action. The algorithm relies on an \emph{action value function}, $Q(s,a)$, which computes the quality of a state-action combination. In Q-learning, an agent tries to find the optimal policy that leads to the maximum action value. Q-learning updates the action value function using an iterative process as shown in Equation~\ref{eq_qLearning}. 
\begin{equation} \label{eq_qLearning}
Q^{new}_{t}(s,a) = (1 - \alpha)Q_{t}(s,a) + \alpha(r_{t+1} + \gamma \underset{a} max Q(s_{t+1},a))
\end{equation}
where $s$ is the current state, $a$ is a specific action, $s_{t+1}$ is the next state as a result of the action, $\underset{a} max Q(s_{t+1},a)$ is the estimated optimal action value in the next state, value $r_{t+1}$ is an observed reward at the next state, $\alpha$ is a learning rate and $\gamma$ is a discount factor. In CLLA, the states, actions and rewards used by the learning agents are defined as follows.

\subsubsection{States} \label{States}
A learning agent can work with four types of states. The first state represents the current traffic signal phase at an intersection. The second state represents the queue length of incoming vehicles that are going to pass the intersection without turning. The third state represents the queue length of incoming vehicles which are going to turn at the intersection. The fourth state represents the queue length of outgoing vehicles, i.e., the vehicles that have passed the intersection. Although it is possible to add other types of states, we find that the combination of the four states can work well for traffic optimization.

\subsubsection{Actions}
There are three possible actions: increasing the number of upstream lanes by 1, increasing the number of downstream lanes by 1 or keeping the current configuration. When the number of lanes in one direction is increased, the number of lanes in the opposite direction is decreased at the same time. Since a learning agent controls a specific road intersection, the agent determines the action for each individual road segment that connects with the intersection. An agent is allowed to take a lane-changing action only when there is a traffic imbalance on the road segment (see Equation~\ref{definition_imbalance} for the definition of traffic imbalance).

\subsubsection{Rewards}
We define the rewards based on two factors. The first factor is the waiting time of vehicles at an intersection. When the waiting time decreases, there is generally an improvement of traffic efficiency. Hence the rewards should consider the difference between the current waiting time and the updated waiting time of all the vehicles that are approaching the intersection. The second factor is the difference between the length of vehicle queues at different approaches to an intersection. When the queue length of one approaching road is significantly longer than the queue length of another approaching road, there is a higher chance that the traffic becomes congested in the former case. Therefore we need to penalize the actions that increase the difference between the longest queue length and the shortest queue length. The following reward function combines the two factors. A parameter $\beta$ is used to give weights for the two factors. We normalized the two factors to stabilize the learning process by limiting reward function between 1 to -1. To give equal priority to both factors, we set $\beta$ to 0.5 in the experiments. 

\begin{displaymath}  \begin{aligned} R = (1-\beta)\times\frac{\text{Current wait time} - \text{Next wait time}}{\text{max(Next wait time, Current wait time)}}\end{aligned} \end{displaymath}
\begin{displaymath} \begin{aligned} - \beta\times\frac{\text{Queue length difference}}{\text{Aggregated road capacity}} \end{aligned} \end{displaymath}

\subsection{Coordination of Lane-direction Changes}
\label{sec-coord}
We develop the coordinating process based on the observation that a locally optimized lane-direction change may conflict with the lane-direction changes that happen in the surrounding areas. A conflict can happen due to the fact that the effect of a lane-direction change can spread from one road segment to other road segments. For example, let us assume that a constant portion of the upstream traffic that passes through road segment $x$ will also pass through road segment $y$ in the upstream direction later on. An increase of the upstream lanes on $x$ can lead to a significant increase of upstream traffic on $x$ due to the increased traffic capacity in the direction. Over time, the traffic volume change on $x$ can lead to the increase of the upstream traffic on $y$, which implies that the number of upstream lanes at $y$ may need to be increased to suit the change of traffic volume. In this case, the lane-direction change at $y$ can be seen as a \emph{consequential change} caused by the change at $x$. However, the learning agent that controls the lane-directions at $y$ may suggest an increase of downstream lanes based on the current local traffic condition at $y$. If this is the case, the locally optimized change will conflict with the consequential change. The key task of the coordinating process is quantifying such conflicts in road networks. If there are a large number of conflicts at a road segment, the locally optimized change should not be applied because it may have a negative impact on the traffic flows at the global level later on. This is a key idea behind the coordination process of our solution. As shown in Section~\ref{sec-CLLA}, our solution applies a proposed lane-direction change to a road segment only when the number of the conflicts is below a given threshold. 

To help identify the conflicts between lane-direction changes, we develop a novel data structure, named \emph{Path Dependency Graph (PDG)}. The data structure maintains several types of traffic information, including the path of traffic flow, the proposed lane-direction changes and the current traffic conditions. The coordinating agent uses PDG to search for the locations of consequential lane-direction changes. The conflicts between lane-direction changes are then identified by comparing the consequential lane-direction changes and the proposed lane-direction changes at the same locations. The coordinating agent also proposes additional lane-direction changes using PDG.

A PDG ($PDG(V^{PDG}, E^{PDG})$)  consists of a number of vertices and a number of directional edges. A vertex $v \in V^{PDG}$ represents a road segment. The corresponding road segments of the two vertices must appear in the path of a vehicle. A vertex can connect to a number of out-degree edges and a number of in-degree edges. The direction of an edge depends on the order of traffic flow that goes through the two road segments. An edge that starts from vertex $v_1$ and ends in vertex $v_2$ shows that the traffic flow will pass through $v_1$'s corresponding road segment first then pass through $v_2$'s corresponding road segment later. We should also note that the two road segments, which are linked by an edge, do not have to share a common road intersection, i.e., they can be disjoint. Given the path of all the vehicles, a PDG can be constructed such that all the unique road segments on the vehicle paths have corresponding vertices in the graph. For each pair of the road segments on a vehicle path, there is a corresponding edge in the graph. If the path of two or more vehicles goes through the same pair of road segments, there is only one corresponding edge in the graph. 

A vertex of PDG has the following properties. 
\begin{itemize}
\item\textbf{Proposed Change:} The proposed lane-direction change at the corresponding road segment. This may be submitted from a learning agent or created by the system during the coordinating process. The property value can be $1$, $0$ and $-1$. A value of $1$ means the upstream direction gets one more lane. A value of $0$ means there is no need for a change. A value of $-1$ means the downstream direction gets one more lane.
\item\textbf{Consequential Changes:} A list of potential lane-direction changes caused by lane-direction changes at other road segments. Similar to the \textbf{Proposed Change} property, the value of a consequential change can be $1$, $0$ and $-1$.
\item\textbf{Imbalance:} The lane direction which has a considerably higher traffic load than the other direction. The property value can be $upstream$, $downstream$ and $none$. In our implementation, the imbalance of traffic load is measured based on the queue length in the opposite directions. Let $q_{up}$ be the upstream queue length and $q_{down}$ be the downstream queue length. Let the total queue length in both directions be $q_{total}$. Let $P$ be a threshold percentage. The property value is computed as follows.
\begin{equation}
    Imbalance = 
    \begin{cases}
      upstream, & \text{if}\ q_{up}/q_{total} > P \\
      downstream, & \text{if}\ q_{down}/q_{total} > P \\
      0, & \text{otherwise}
    \end{cases}
    \label{definition_imbalance}
\end{equation}
Due to the dynamic nature of traffic, imbalance value may change frequently, leading to frequent changes of lane-directions. This may not be ideal in practice. One can get a steady imbalance value by adding certain restrictions in the computation. For example, one may require that the ratio between upstream queue length and the total queue length must be above the threshold for a certain period of time before setting the imbalance value to upstream.
\item\textbf{Current Lane Configuration:} The number of upstream lanes and the number of downstream lanes in the corresponding road segment.
\end{itemize}

An edge of PDG has a property called \textbf{impact}, which shows whether a lane-direction change at the starting vertex can lead to the same change at the ending vertex. The value of this property can be $1$ or $-1$. A value of $1$ means the change at both vertices will be the same. For example, if the change at the starting vertex is increasing the number of upstream lanes, the change at the ending vertex will also be increasing the number of upstream lanes. A value of $-1$ means the changes at the vertices will be opposite to each other. The relationship between the changes and the property value is shown in Equation~\ref{equ-impact-change}, where the starting vertex is $v_1$ and the ending vertex is $v_2$. The property value is determined based on the path of the majority of the vehicles that move between the two corresponding road segments. If the path passes through both road segments in the same direction, the property value is $1$. Otherwise, the property value is $-1$. The impact property is key for finding the consequential change at the ending vertex given the change at the starting vertex. As shown in Equation~\ref{equ-consequential}, the consequential change at the ending vertex can be computed based on the property value and the initial change at the starting vertex.
\begin{equation} \label{equ-impact-change}
    impact_{(v_1,v_2)} = change_{v_1} \times change_{v_2}  
\end{equation}
\begin{equation} \label{equ-consequential}
    change_{v_2} = impact_{(v_1,v_2)} \times  change_{v_1}
\end{equation}
When constructing a PDG, it may not be necessary to consider the full path of vehicles due to two reasons. First, the full path of vehicles can consist of a large number of road segments. The size of the graph can grow exponentially when the length of path increases. Second, due to the highly dynamic nature of traffic, the coordination of lane-direction changes should only consider the traffic conditions in the near future. Therefore, in our implementation, we set an upper limit to the number of road segments in vehicle paths when building a PDG. The limit is called \emph{lookup distance} in our experiments.

We show an example road network (Figure~\ref{roadnetwork}) and its corresponding PDG (Figure~\ref{pdg}). The road network has 12 roads segments (A to L). There are two paths going through the network, path $\alpha$ and path $\beta$. Path $\alpha$ passes through 4 edges (A, F, I, J). Path $\beta$ passes through 3 edges (C, F, H). The 7 edges correspond to 7 vertices in the PDG. The PDG contains 3 edges starting from A (A-F, A-I, A-J) because path $\alpha$ passes through A, F, I and J in the road network. Similarly, the PDG contains 2 edges that starting from F. 

For each edge in the PDG, the value of its impact property is attached to the edge. As path $\alpha$ goes through the edges (A, F, I and J) in the upstream direction (Figure~\ref{roadnetwork}), the impact value at all the edges between the corresponding vertices is 1 in the PDG (Figure~\ref{pdg}). Differently, the impact value of the edge C-H is -1 in the PDG. This is because path $\beta$ goes through C in the upstream direction but it goes through H in the downstream direction.

\begin{figure}[htbp]
    \begin{subfigure}[b]{.28\textwidth}
    \centering
	\includegraphics[width=.88\linewidth]{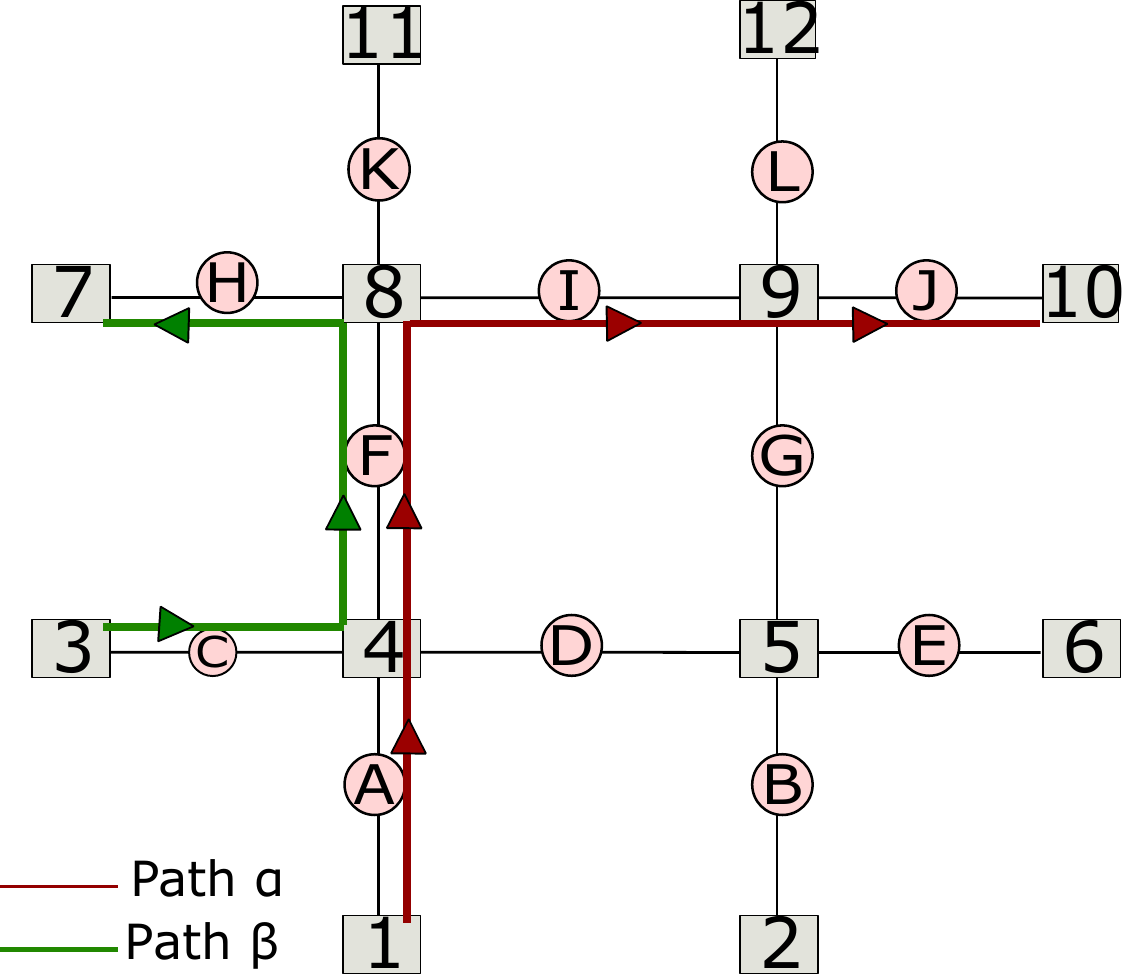}
	\caption{A simple road network with\\ two paths (red and green)} 
	\label{roadnetwork}
	\end{subfigure}%
    \begin{subfigure}[b]{.22\textwidth}
    \centering
	\includegraphics[width=.88\linewidth]{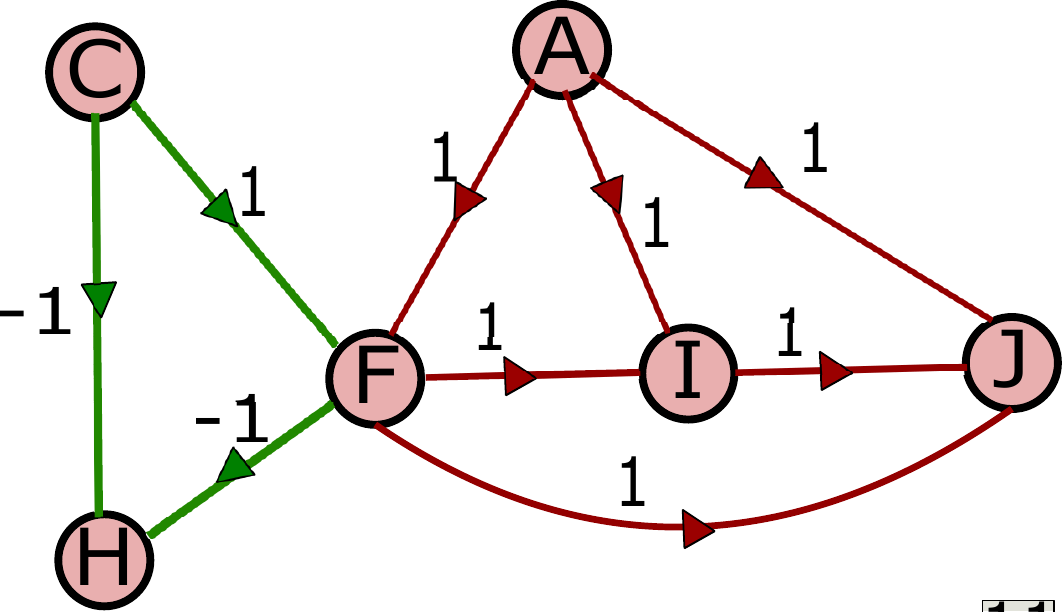}
	\caption{Path dependency graph (PDG) based on the road network in Figure~\ref{roadnetwork}}.
	\label{pdg}
	\end{subfigure}
\vspace{-6mm}
\end{figure}

The coordinator uses the Direction-Change Evaluation algorithm (Algorithm~\ref{PDG Algorithm}) to quantify the conflicts between lane-direction changes. The algorithm traverses through a PDG in a breadth-first manner in iterations. The number of iterations is controlled by a \emph{depth} parameter (shown as $dp$ in Algorithm~\ref{PDG Algorithm}). In the first round of iteration, the algorithm starts with the lane-direction changes that are proposed by the bottom-layer learning agents. For each vertex with a proposed change, its first-depth neighbours (out-degree nodes) are visited (Step 6). For each of the neighbours, the consequential change caused by the proposed change is computed. This can be done with the process shown in Equation~\ref{equ-consequential}. Then the algorithm updates the count of conflicts at the neighbour's corresponding road network edge. In the next iteration, the algorithm starts with all the neighbour vertices that are visited in the previous round. After each iteration $dp$ is decremented. The algorithm stops when $dp$ reaches zero.

The Direction-Change Evaluation algorithm not only quantifies the conflicts between lane-direction changes but also expands the set of lane-direction changes for the road segments that are visited during the traversal of the PDG. The rationale is that the bottom-layer learning agents may not propose lane-direction changes for road segments when they do not predict any benefit of the change based on local traffic conditions. However, the lane-direction changes in other parts of the road network may eventually affect the traffic conditions at these road segments, leading to traffic congestions. The algorithm pre-emptively attempts lane-direction changes for these road segments when it predicts that there can be consequential changes caused by the changes in other parts of the road network. This can help mitigate incoming traffic congestions. As shown in Step 6 of Algorithm~\ref{PDG Algorithm}, a Direction-Change Creation algorithm is used for proposing additional lane-direction changes. Details of the Direction-Change Creation algorithm are shown in Algorithm~\ref{Road decision algorithm}. Every time coordinator executes \textbf{Direction-Change Evaluation}, a new lane configuration is computed and new $G^\prime_{t}(V,E^\prime)$ is generated.

\begin{algorithm}[t]
	
	\KwIn{$EC_{initial}$, a set of edge-change pairs proposed by the learning agents}	
	\KwIn{$G(V,E)$, A road network graph. Each edge in the graph has a property, \textbf{conflict count}, which has an integer value that is set to 0 initially.}
	\KwIn{$l$, the lookup distance of PDG}
	\KwIn{$dp$, the depth of search}
	\KwOut{$EC_{expanded}$, a set of edge-change pairs given by the coordinator}
	
	Build a PDG based on the next $l$ road segments on the path of vehicles. For each PDG vertex, its properties, \textbf{proposed change} and \textbf{consequential changes}, are set to empty values initially.
	
	Create an empty set $N$. For each edge-change pair in $EC_{initial}$, find the corresponding vertex $v$ in PDG and update its \textbf{proposed change} property. Add $v$ to $N$. 
	
	Set the current depth of search to $dp$. 
	
	If the current depth is above 0, do the following steps. Otherwise, jump to Step 8.
	
	Create an empty set $N'$. 	
	
	For each $v$ in $N$, first check whether $v$ has a proposed change. If not, get a proposed change for $v$ using the \textbf{Direction-Change Creation} algorithm. Then for each of $v$'s neighbours at the end of its out-degree arcs, $v_o$, identify the consequential change at the vertex that is caused by the proposed change at $v$. Add the consequential change to the \textbf{consequential changes} of $v_o$ if the change does not exist on the list. If $v_o$ already has a proposed change but the proposed change is different to the consequential change at $v_o$, increase the \textbf{conflict count} of the corresponding road network edge by 1. Add $v_o$ to $N'$.
	
	Decrease the current depth of search by 1. Replace the vertices in $N$ with the vertices in $N'$. Go back to Step 4.	
	
	For each PDG vertex $v$ with a proposed change, create a corresponding edge-change pair and add the pair to $EA_{expanded}$. Exit the algorithm.
	
	\caption{\bf Direction-Change Evaluation} 
	\label{PDG Algorithm}
\end{algorithm}

\begin{algorithm} [t]
	\KwIn{$v$, a PDG vertex that corresponds to an edge in a road network graph. The value of the \textbf{imbalance} property is set to $none$ initially.}
	\KwOut{$change$, the proposed lane-direction change for $e$, which can be $1 (upstream)$, $0 (none)$ and $-1 (downstream)$. The default value is $0$.}
	
	$consequential_{up}$: whether the consequential changes at $v$ include one that increases the number of upstream lanes.
	
	$consequential_{down}$: whether the consequential changes at $v$ include one that increases the number of downstream lanes.
	
	\If{\textbf{imbalance} = upstream}{
	    \If {$consequential_{up} = True$ and $consequential_{down} = False$}{
	        $change \gets 1$ (change one lane from downstream to upstream)}
	 }
	 
	 \If{\textbf{imbalance} = downstream}{
	    \If {$consequential_{up} = False$ and $consequential_{down} = True$}{
	        $change \gets -1$ (change one lane from upstream to downstream)}
	 } 
	 
	 \If{\textbf{imbalance} = none}{
	    \If {($consequential_{up} = True $ and downstream has more lanes than upstream}
	        {$change \gets 1 $ (change one lane from downstream to upstream)}
	    \If {($consequential_{down} = True $ and upstream has more lanes than downstream}
	        {$change \gets -1$ (change one lane from upstream to downstream)}
	  }
	        
    \caption{\bf Direction-Change Creation} \label{Road decision algorithm}
\end{algorithm}

\textbf{Complexity of Coordinating Process.}
Let us assume there are $m$ number of requests from bottom layer agents. The degree of a node in PDG is $deg(v)$, where $v \in V^{PDG}$. The algorithm traverses $m$ BFTs throughout the PDG for certain depth $dp$. Then complexity of $m$ BFTs is $\mathcal{O}(m(deg(v)^{dp}))$. However, according to lemma \ref{deg lemma}, $deg(v)$ is independent of the road network size or number of paths for a given $l$. The depth $dp$ is constant irrespective of the road network size. Then the algorithm complexity can be reduced to $\mathcal{O}(m)$; the algorithm complexity is linear in the number of requests from agents within the buffering period. In the worst case, there can be requests for each road segment of the road network, $G(V,E)$, leading to the complexity of $\mathcal{O}(|E|)$.

\textbf{Distributed version.} Algorithm \ref{PDG Algorithm} can work with a set of distributed agents (coordinating agents) in the upper layer. In Algorithm \ref{PDG Algorithm}, execution is independent of the order of requests coming from the bottom layer agents. Therefore, requests can be processed in a distributed manner. Every coordinating agent traverses first depth and inform changes to other agents. Once every agent finishes their first depth, all coordinating agents start their second depth, and so on. In such a setting, the complexity of the algorithm is $\mathcal{O}(1)$. In this work, we implemented the centralized version, however, when applied to larger road networks, the distributed version can be implemented.

\section{Experimental methodology} 
\label{experiments}

We compare the proposed algorithm, CLLA, against DLA and two other baseline algorithms using traffic simulations. We evaluate the performance of the algorithms for two road networks based on real traffic data. The effects of individual parameters of CLLA and DLA are also evaluated. The rest of the section details the settings of the experiments. 

\subsection{Experimental setup}
\textbf{Simulation Model.} We simulate a traffic system similar to the ones used for reinforcement learning-based traffic optimization in~\cite{Wiering,kohler2018traffic}. In our implementation, vehicles on a road link are modelled based on travel time, which is the sum of two values, \emph{pure transmit time} and \emph{waiting time}. \emph{Pure transmit time} is the time taken by a vehicle to travel through the road link at the free-flow speed. \emph{Waiting time} is the duration that a vehicle waits in a traffic signal queue. When the direction of a lane needs to be changed, all existing vehicles in the lane need to leave the lane and move into the adjacent lane in the same direction. Vehicles travelling in the opposite direction can use the lane only after it is cleared of traffic.

\textbf{Road Networks.} We run experiments based on the real taxi trip data from New York City~\cite{nyc}. The data includes the source, the destination and the start time of the taxi trips in the city. We pick two areas for simulation (Figure~\ref{NW1} and Figure~\ref{NW2}) because the areas contain a larger number of sources and destinations than other areas. 
\begin{figure}
	\vspace{-1.5em}
	\begin{subfigure}[b]{.25\textwidth}
		\centering
		\includegraphics[width=0.98\linewidth]{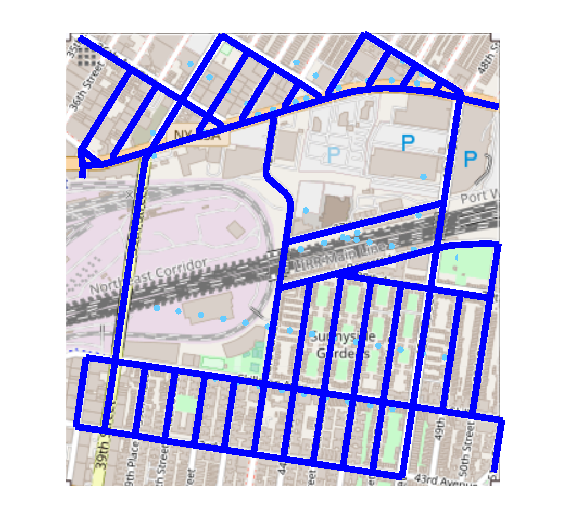}
		\caption{Long Island (LI)}
		\label{map1}
	\end{subfigure}%
	\begin{subfigure}[b]{.25\textwidth}
		\centering
		\includegraphics[width=1\linewidth]{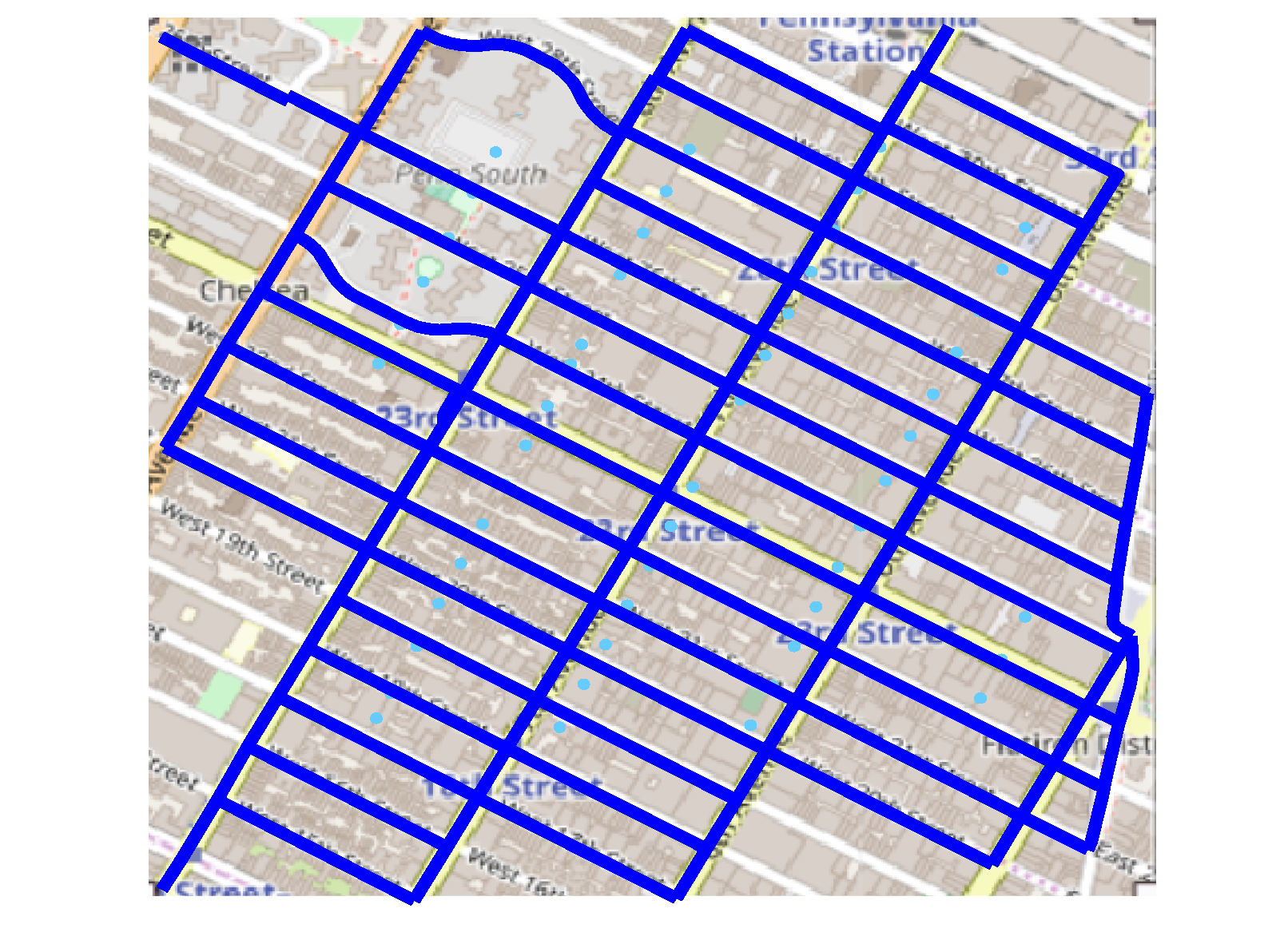}
		\caption{Midtown Manhattan (MM)}
		\label{map2}
	\end{subfigure}
	\caption{The road network of two simulation areas in New York}
\end{figure}
The road network of the simulation areas is loaded from OpenStreetMap~\cite{osm}. For a specific taxi trip, the source and the destination are mapped to the nearest OpenStreetMap nodes. The shortest path between the source and destination is calculated. The simulated vehicles follow the shortest paths generated from the taxi trip data. 

\textbf{Comparison baselines.}
\label{sec-baseline}
Different to the proposed solution, CLLA, the existing approaches for optimizing lane-directions are based on linear programming, which makes them unsuitable for large-scale dynamic optimization due to the high computation cost. Due to the lack of comparable solutions, we define three baseline solutions, which are used to compare against CLLA. In our experiments, the traffic signals use static timing and phasing, regardless of which solution is used. We conduct comparative tests against the following solutions:

\begin{itemize}
	\item {\textbf{No Lane-direction Allocations}} (\textbf{no-LA}): This solution does not do any lane-direction change. The traffic is controlled by static traffic signals only.
	\item {\textbf{Demand-based Lane Allocations}} (\textbf{DLA}): In this solution, the lane-direction changes are computed with Algorithm~\ref{baseline_algo}.
	\item {\textbf{Local Lane-direction Allocations}} (\textbf{LLA}): This solution uses multiple learning agents to decide lane-direction changes. The optimization is performed using the approach described in Section~\ref{sec-qLearning}. LLA is similar to CLLA but there is no coordination between the agents.
	\item {\textbf{Coordinated Learning-based Lane Allocations}} (\textbf{CLLA}): This is the two layer optimization framework described in Section~\ref{sec-CLLA}. 
\end{itemize}

\subsection{Evaluation Metrics}
We measure the performance of the solutions based on the following metrics.

\textbf{Average travel time}:
The travel time of a vehicle is the duration that the vehicle spends on travelling from its source to its destination. We compute the average travel time based on all the vehicles that complete their trips during a simulation. A higher average travel time indicates that the traffic is more congested during the simulation. Our proposed solutions aim to reduce the average travel time. More information about this metric is shown in Section~\ref{prob def}.

\textbf{Deviation from free-flow travel time}: 
The free-flow travel time of a vehicle is the shortest possible travel time, achieved when the vehicle travels at the speed limit of the roads without slowing down at traffic lights during its entire trip. Deviation from Free-Flow travel Time ($DFFT$) is defined as in Equation~\ref{equ-dfft}, where $t_a$ is the actual time and $t_f$ is the free-flow travel time. The lowest value of DFFT is 1, which is also the best value that a vehicle can achieve.

\begin{equation}\label{equ-dfft}
DFFT = t_a / t_f
\end{equation}

\subsection{Parameter Sensitivity Testing}
We evaluate the effects of the hyper-parameters of \textbf{CCLA} and \textbf{DLA}, which are directly related to lane-direction changes in the simulation model. To evaluate the effects of a specific parameter, we run a set of tests varying the value of the parameter (while keeping the value of other parameters at their default reported in Table~\ref{tbl-sensitivitySettings}). The average travel time is reported for each of the tests. The detailed settings of the parameters are shown in Table~\ref{tbl-sensitivitySettings}. We describe the parameters as follows.

\begin{table}[tbp]
	\centering
	\begin{tabular}{|l|l|l|}
		\hline
		\textbf{Parameter}  & \textbf{Range}    & \textbf{\begin{tabular}[c]{@{}l@{}}Default \\ value\end{tabular}} \\ \hline
		Cost of lane-direction change in CLLA&&\\
		and DLA (seconds)                                                              & 40 - 480 & 120                                                      \\ \hline
		Aggressiveness of lane-direction change &&\\
		in CLLA (seconds)                                          & 100-1000 & 300                                                      \\ \hline
		Depth in CLLA                                                                                & 1-5      & 2                                                        \\ \hline
		Lookup distance in CLLA                                                                      & 3-7      & 5                                                        \\ \hline
		Update period in CLLA (minutes)                                                              & 0.3 - 20 & 0.3 \\ \hline
		Update period in DLA (minutes) & 2.5-20     & 10                                                    \\ \hline
	\end{tabular}
	\caption{Settings used in the parameter sensitivity experiments}
	\label{tbl-sensitivitySettings}
\end{table}

\textbf{Cost of lane-direction change in CLLA}: The cost of a lane-direction change is the time spent on clearing the lane that needs to be changed. When the direction of a lane needs to be changed, all the existing vehicles in the lane need to leave the lane before the lane can be used by the vehicles from the opposite direction. The time spent on clearing the lane can vary due to various random factors in the real world. For example, the vehicles in the lane may not be able to move to an adjacent lane immediately if the adjacent lane is highly congested. We vary the value of this parameter in a large range, from 40 seconds to 480 seconds. 

\textbf{Aggressiveness of lane-direction change in CLLA}: This parameter affects the minimum interval between lane-direction changes. A lane-direction change can only happen when there is a traffic imbalance between the two directions at a road segment. The imbalance is computed based on the model as shown in Equation~\ref{definition_imbalance} (Section~\ref{sec-coord}). Based on an existing study~\cite{Lane}, we set the threshold percentage $P$ of the model to 65\% and require that the traffic imbalance must last for a minimum time period before a lane-direction change can be performed. We define the aggressiveness of lane-direction changes in CLLA as the length of the period. When the period is short, the system can perform lane-direction changes at smaller intervals, and vice-versa. 

\textbf{Depth in CLLA}: This is the parameter $dp$ used in Algorithm~\ref{PDG Algorithm}. When the depth is larger, CLLA can explore more vertices in the PDG, which allows it to detect the impact of a lane-direction change on the road segments that are further away from the location of the change. 

\textbf{Lookup distance in CLLA}: This is the parameter $l$ used in Algorithm~\ref{PDG Algorithm}. It can affect the number of vertices and the number of edges in a PDG. With a higher lookup distance, the PDG needs to consider more road segments in the path of vehicles, which can help identify the impact of lane-direction changes at a longer distance but can increase the size of the graph at the same time.

\textbf{Update period in CLLA}: This parameter controls the frequency at which coordinating agents decide on lane-direction changes. CLLA is suitable for highly dynamic traffic environments. Hence the update period $\Delta t$ can be set to a low value. We vary the value of this parameter between 0.3 minute to 20 minutes with the default value set to 0.3 minute.  

\textbf{Update period in DLA}: This parameter affects the frequency at which DLA optimizes lane-direction changes. DLA decides on lane-direction changes based on the traffic demand that is collected within the update period $\Delta t$ prior to the optimization process. We vary the value of this parameter between 2.5 min to 20 min with the default value set to 10.

\section{Experimental Results}
We now present experimental results when comparing CLLA against the baseline algorithms in the first part, and present the sensitivity analysis to the parameter values of the algorithms in the second part. 

\subsection{Comparison against the baselines}

This experiment compares the performance of the four solutions, which are described in Section~\ref{sec-baseline}. We run a number of simulations in this experiment. For each simulation, we extract taxi trip information for one hour using the real taxi trip data from New York. Based on the real data, we generate traffic in the simulation. The experiment is done for two areas as shown in Figure~\ref{map1} and Figure~\ref{map2}. To simulate a larger variety of traffic scenarios, we also up-sample the trip data to generate more vehicles. We define an \emph{Up Sampled Factor}, which is the number of vehicles that are generated based on each taxi trip in the taxi data. 

For LLA and CLLA, the learning rate $\alpha$ is 0.001 and the discount factor used by Q-learning is 0.75. The parameter $minLoad$ of DLA is set to 100. For other parameters of the solutions, we use the default values as shown in Table~\ref{tbl-sensitivitySettings}. 

\textbf{Average travel time: }
Figure~\ref{NW1} and Figure~\ref{NW2} show the average travel time achieved with the four solutions. \textbf{CLLA} outperforms the other solutions in both simulation areas. We can observe that the average travel time of \textbf{LLA} and \textbf{CLLA} is significantly lower compared to the average travel time of \textbf{no-LA}, which shows the benefit of dynamic lane-direction changes. Although \textbf{DLA} achieves lower travel times than \textbf{no-LA}, it does not perform well compared to \textbf{CLLA} for both areas. \textbf{CLLA} performs consistently better than \textbf{LLA}, because \textbf{LLA} only makes lane-direction changes based on local traffic information without coordination.

We also test the performance of the solutions for a different scenario, where the traffic demand is static.  Vehicles are generated at a constant rate during a 30-minute period. Under this setting, the traffic is less dynamic than the previous scenario, where the traffic demand is based on real data. Figure~\ref{NW1:uniform demand} and Figure~\ref{NW2:uniform demand} show the average travel time achieved with the four solutions. Interestingly, \textbf{DLA} performs as good as \textbf{CLLA}. This is due to the fact that \textbf{DLA} optimizes traffic based on the estimated traffic demand. As the traffic demand is kept constant, the estimated demand can match the actual demand, resulting in the good performance of \textbf{DLA}. On the other hand, \textbf{CLLA} is developed for highly dynamic traffic environments. When the traffic is static, such as in this scenario, the advantage of the solution is limited. 

The results show that \textbf{DLA} can work well with static traffic but does not work well with highly dynamic traffic. \textbf{CLLA} on the other hand works well in both environments: substantially outperforming the baselines in dynamic environments, and matching the performance of DLA in static environments.

\begin{figure}
	\begin{subfigure}[b]{.25\textwidth}
		\centering
		\includegraphics[width=.88\linewidth, height=3.5cm]{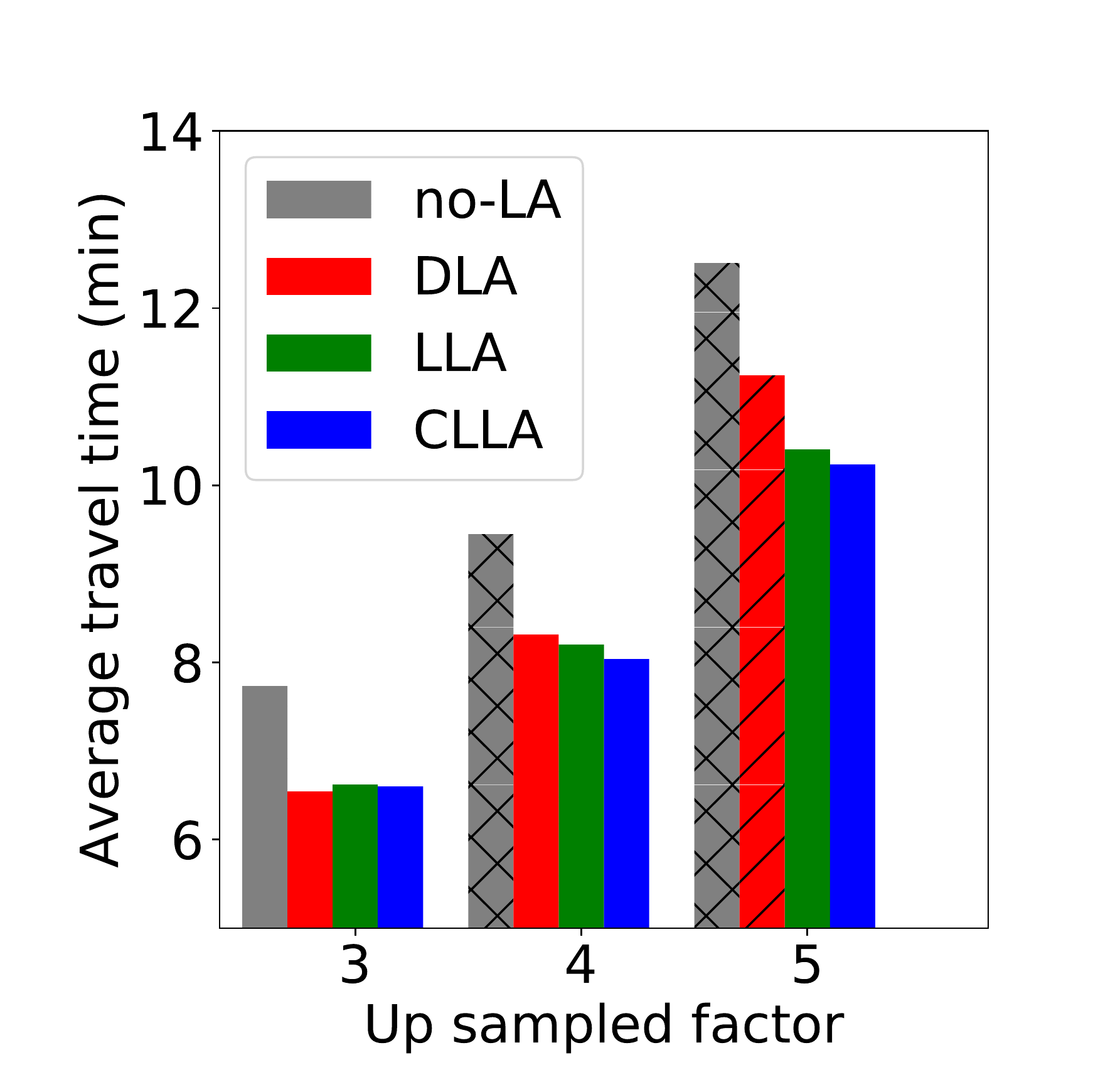}
		\caption{Long Island}
		\label{NW1}
	\end{subfigure}%
	\begin{subfigure}[b]{.25\textwidth}
		\centering
		\includegraphics[width=.88\linewidth, height=3.5cm]{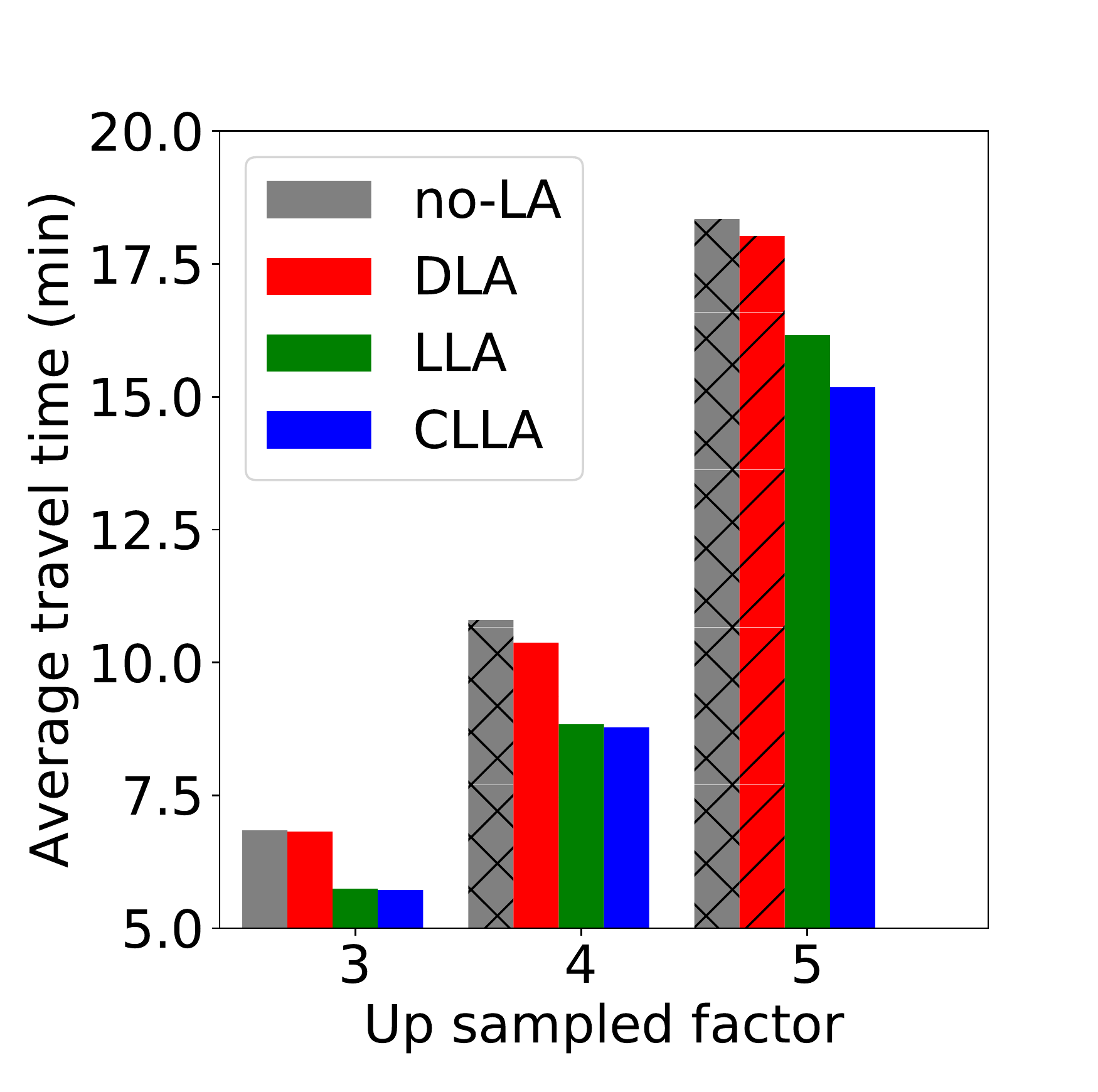}
		\caption{Midtown Manhattan}
		\label{NW2}
	\end{subfigure}
	\caption{Performance of four solutions with dynamic traffic}
	\vspace{-4mm}
\end{figure}

\begin{figure}
	\begin{subfigure}[b]{.25\textwidth}
		\centering
		\includegraphics[width=.88\linewidth, height=3.5cm]{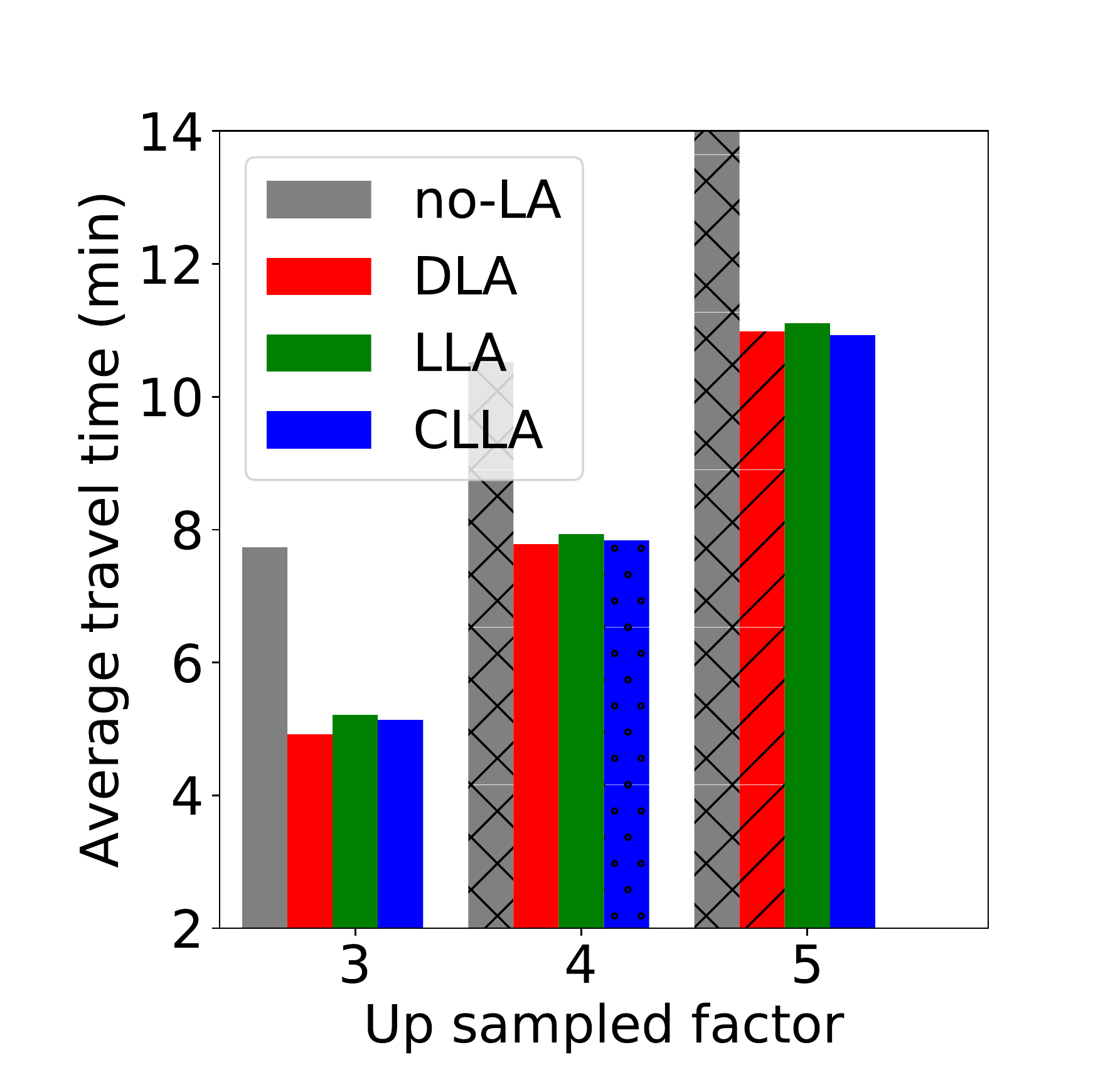}
		\caption{Long Island}
		\label{NW1:uniform demand}
		
	\end{subfigure}%
	\begin{subfigure}[b]{.25\textwidth}
		\centering
		\includegraphics[width=.88\linewidth, , height=3.5cm]{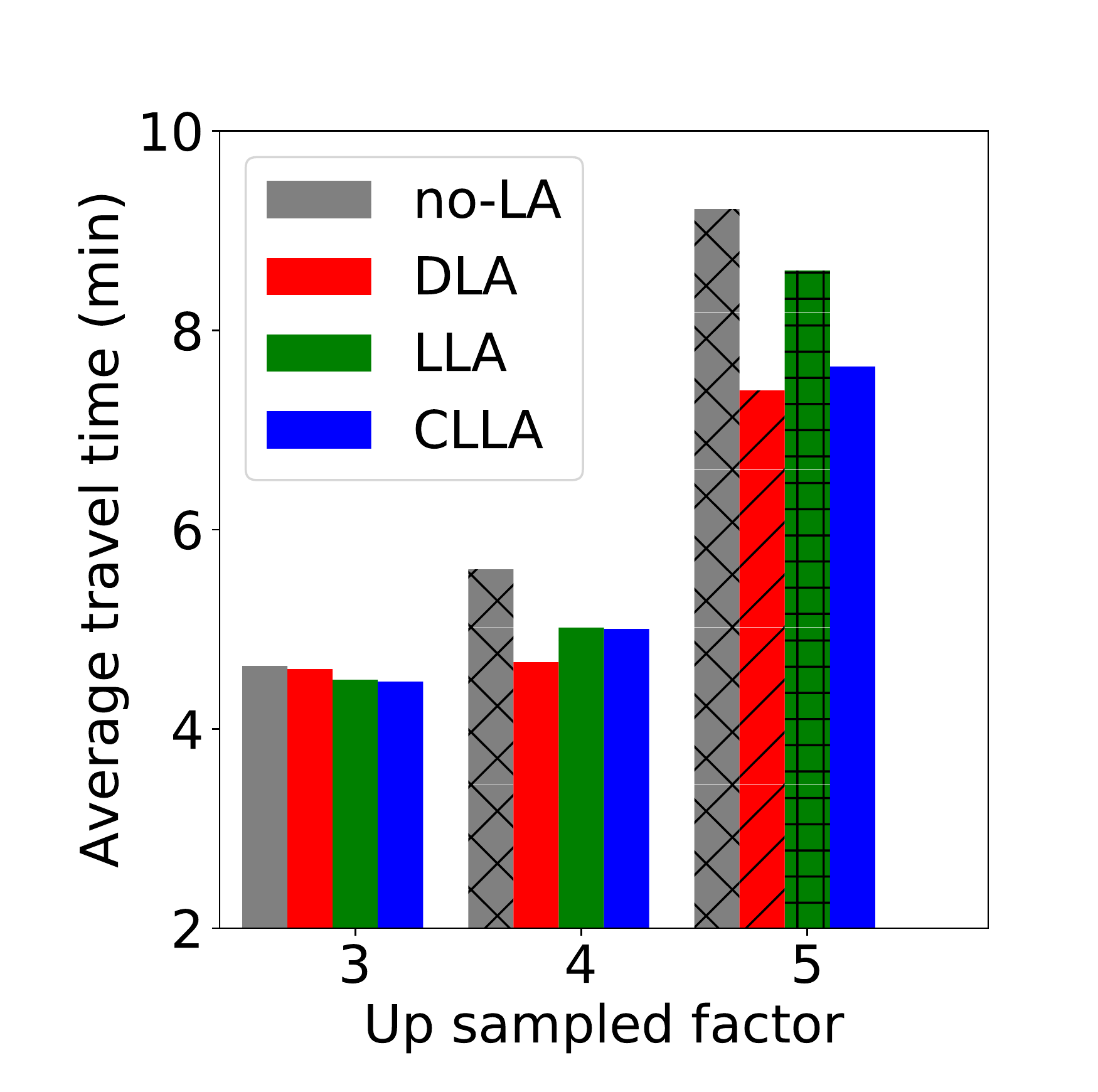}
		\caption{Midtown Manhattan}
		\label{NW2:uniform demand}
	\end{subfigure}
	\caption{Performance of four solutions with static traffic}
	
\end{figure}

\textbf{Deviation from free-flow travel time (DFFT)}: Table~\ref{DFFT_table} shows the percentage of vehicles whose travel time is 10 times or more than their free-flow travel time. The results show that \textbf{LLA} and \textbf{CLLA} are able to achieve a lower deviation from the free-flow travel time compared to \textbf{DLA}. 

\begin{table}[]
	\centering
	\begin{tabular}{|l|l|l|}
		\hline
		Solution & Long Island & Midtown Manhattan\\ \hline
		DLA       & 10.04\%                                                                    & 49.52\%                                                                     \\ \hline
		LLA        & 7.78\%                                                                     & 44.19\%                                                                     \\ \hline
		CLLA       & 7.77\%                                                                     & 46.12\%                                                                     \\ \hline
	\end{tabular}
	\caption{The percentage of vehicles with a DFFT of higher than 10}
	\label{DFFT_table}
	\vspace{-1mm}
\end{table}

\subsection{Parameter sensitivity testing}

For evaluating the effects of individual parameters, we run simulations in the area shown in Figure~\ref{map1}. Each simulation lasts for one hour, during which the traffic is generated based on the real taxi trip data from the area. Figure~\ref{CLLA params} shows the effects of four parameters of \textbf{CLLA}. Figure~\ref{update freq: DLA, CCLA} compares the effects of the update period between \textbf{DLA} and \textbf{CLLA}. 

Our result shows that the travel time increases when the cost of a lane-direction change increases (Figure~\ref{cost of lane change}). The result indicates that lane-direction changes may not be beneficial in all circumstances. When the cost of lane-direction changes is high, performing the changes can cause significant interruption to the traffic and negate the benefit of the changes. 

Figure~\ref{Lane change: Aggressiveness } shows how the aggressiveness of lane-direction changes can affect the travel time of vehicles. The result shows that a low level of aggressiveness and a high level of aggressiveness have a negative impact on travel times. When the level of aggressiveness is low, the lane-direction changes can only happen in large intervals. Hence the changes may not adapt to the dynamic change of traffic. When the level of aggressiveness is high, the system changes the direction of lanes at a high frequency, which can cause significant interruption to the traffic attributed to taking the time to clear the lanes during the changes. 

Our result shows that the best depth for traversing the PDG is 2 (Figure~\ref{depth}). When the depth changes from 1 to 2, we observe a decrease in travel time. However, when the depth is higher than 2, we do not observe a decrease of travel time. When the depth is higher, the system can identify the impact of a lane-direction change that are further away. However, the impact can become negligible if the lane-direction change is far away. This is the reason there is no improvement of travel time when the depth is higher than 2.

Figure~\ref{length} shows that a larger lookup distance can result in a lower average travel time. When the lookup distance increases, \textbf{CLLA} considers more road segments in a vehicle path when building the PDG. This helps identify the consequential lane-direction changes on the same path. The reduction in the average travel time becomes less significant when the lookup distance is higher than 2. This is because the impact of a lane-direction change reduces when the change is further away.

When the update period $\Delta t$ of \textbf{DLA} is below 5 minutes or beyond 15 minutes, it is less likely to get a good estimation of traffic demand, which can lead to a relatively high travel time (Figure~\ref{DLA update}). The average travel time is at its minimum when $\Delta t$ is set to 10 minutes. Different to \textbf{DLA}, the travel time achieved with \textbf{CLLA} grows slowly with the increase of $\Delta t$ until $\Delta t$ reaches beyond 15 minutes. The relatively steady performance of \textbf{CLLA} shows that the coordination between lane-direction changes can help mitigate traffic congestion for a certain period of time in the future. If minimizing the average travel time is of priority, one can set $\Delta t$ to a very low value, e.g., 5 minutes. If one needs to reduce the computation cost of the optimization while achieving a reasonably good travel time, the $\Delta t$ can be set to a larger value, e.g., 15 minutes. 

\begin{figure}[t]
	\begin{subfigure}[b]{.25\textwidth}
		\centering
		\includegraphics[width=.88\linewidth, height=2.5cm]{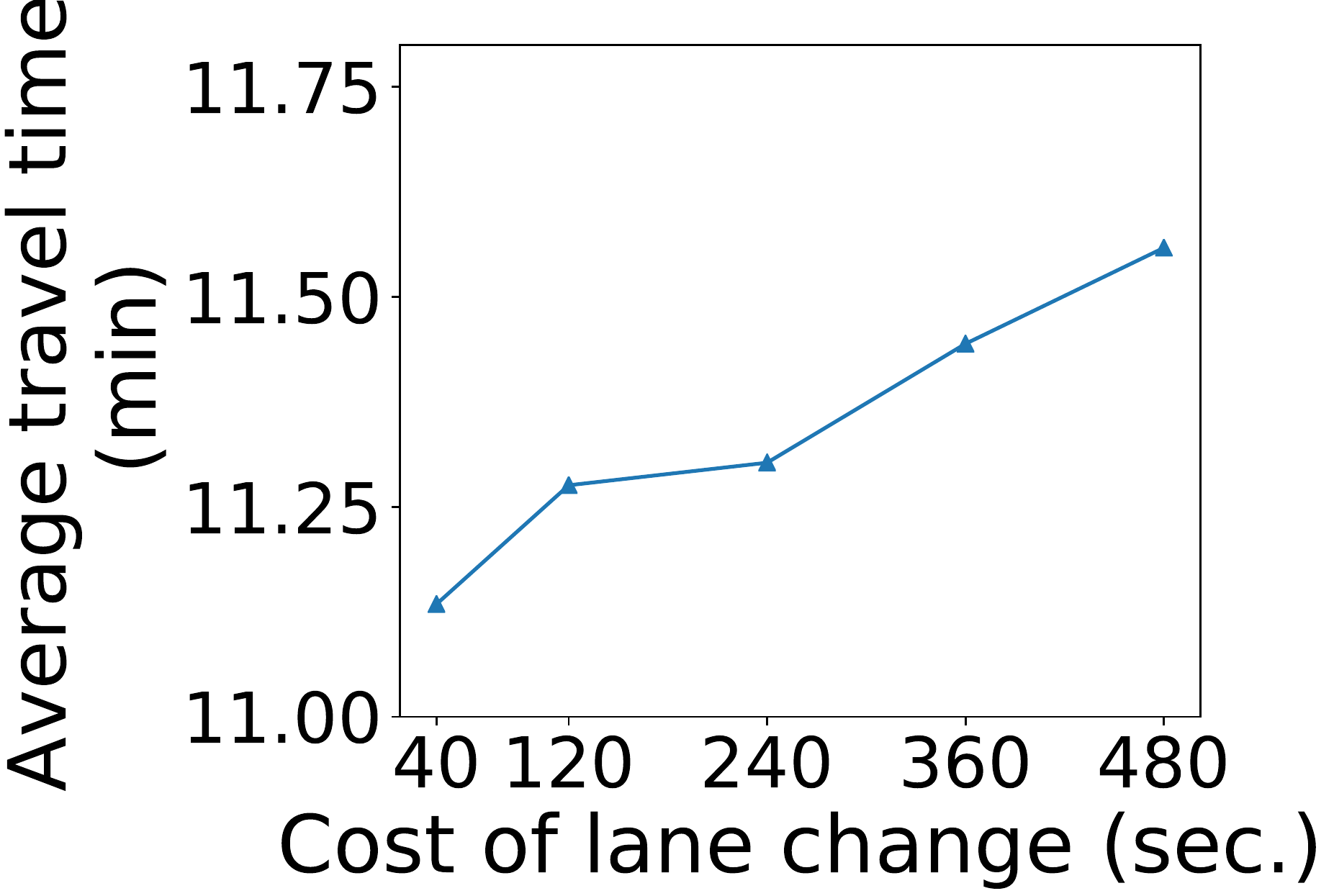}
		\caption{Cost of change}
		\label{cost of lane change}
	\end{subfigure}%
	\begin{subfigure}[b]{.25\textwidth}
		\centering
		\includegraphics[width=.88\linewidth, height=2.5cm]{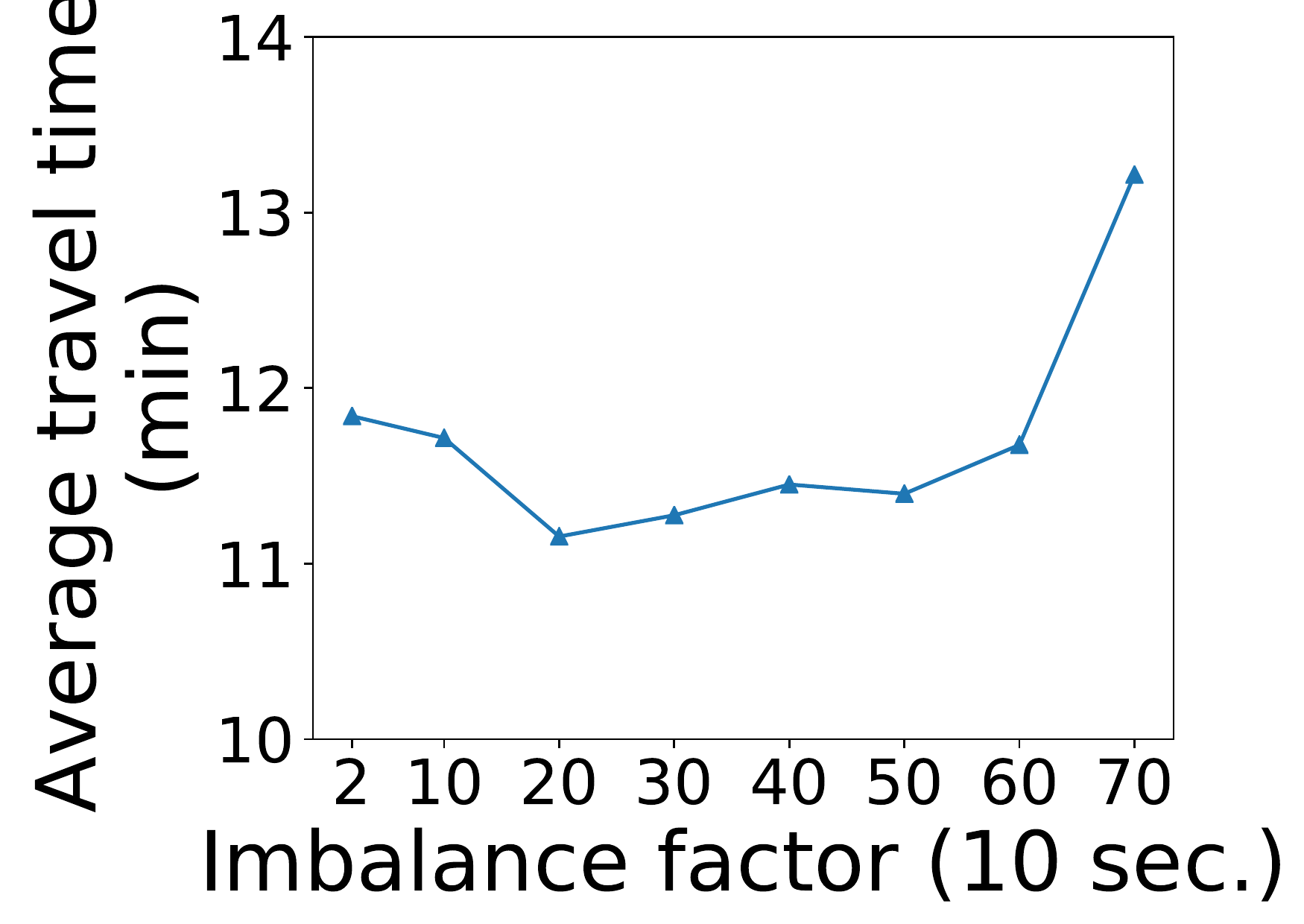}
		\caption{Aggressiveness of change}
		\label{Lane change: Aggressiveness }
	\end{subfigure}
	
	\begin{subfigure}[b]{.25\textwidth}
		\centering
		\includegraphics[width=.88\linewidth, height=2.5cm]{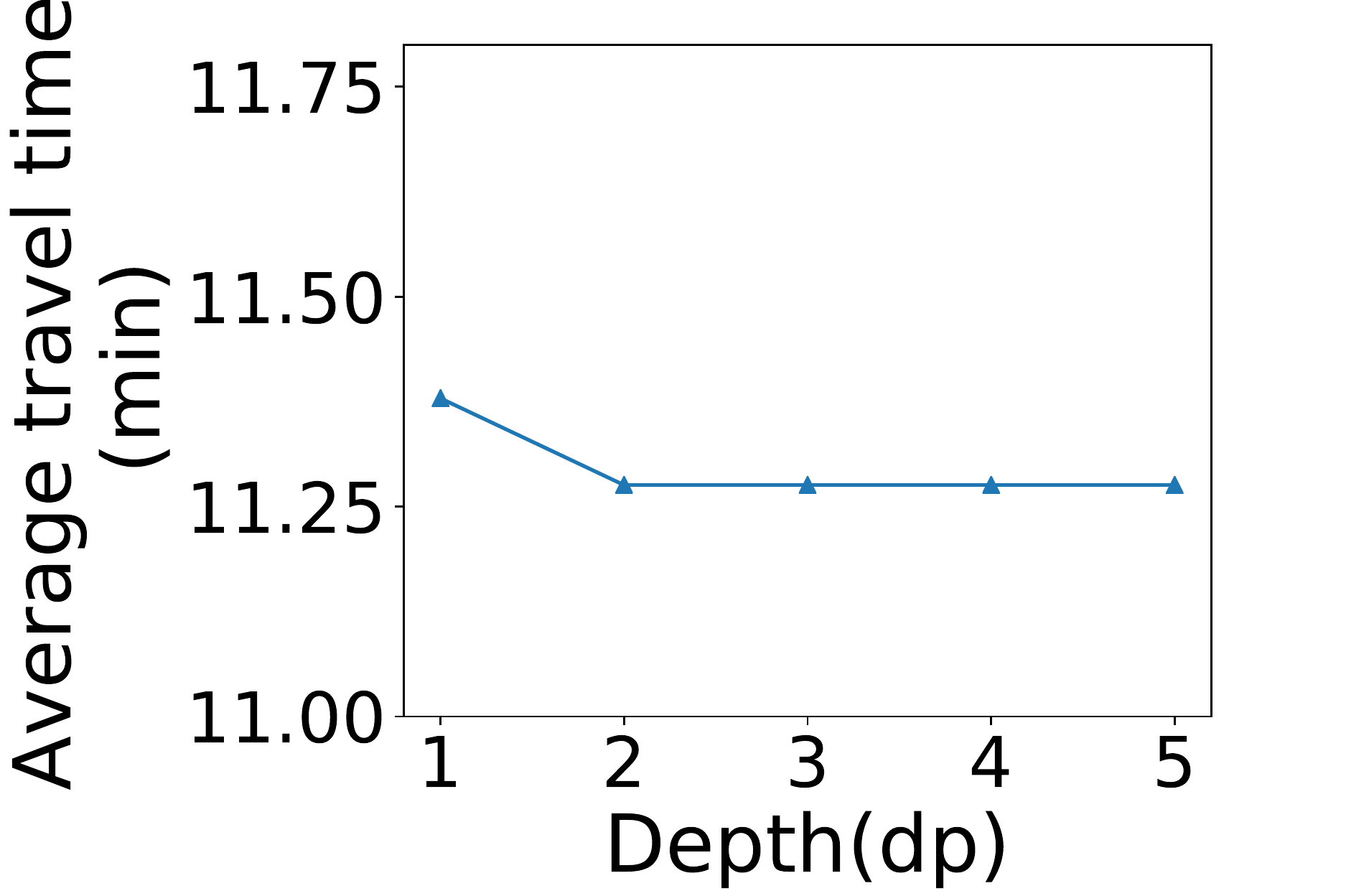}
		\caption{Depth ($dp$)}
		\label{depth}
	\end{subfigure}%
	\begin{subfigure}[b]{.25\textwidth}
		\centering
		\includegraphics[width=.88\linewidth, height=2.5cm]{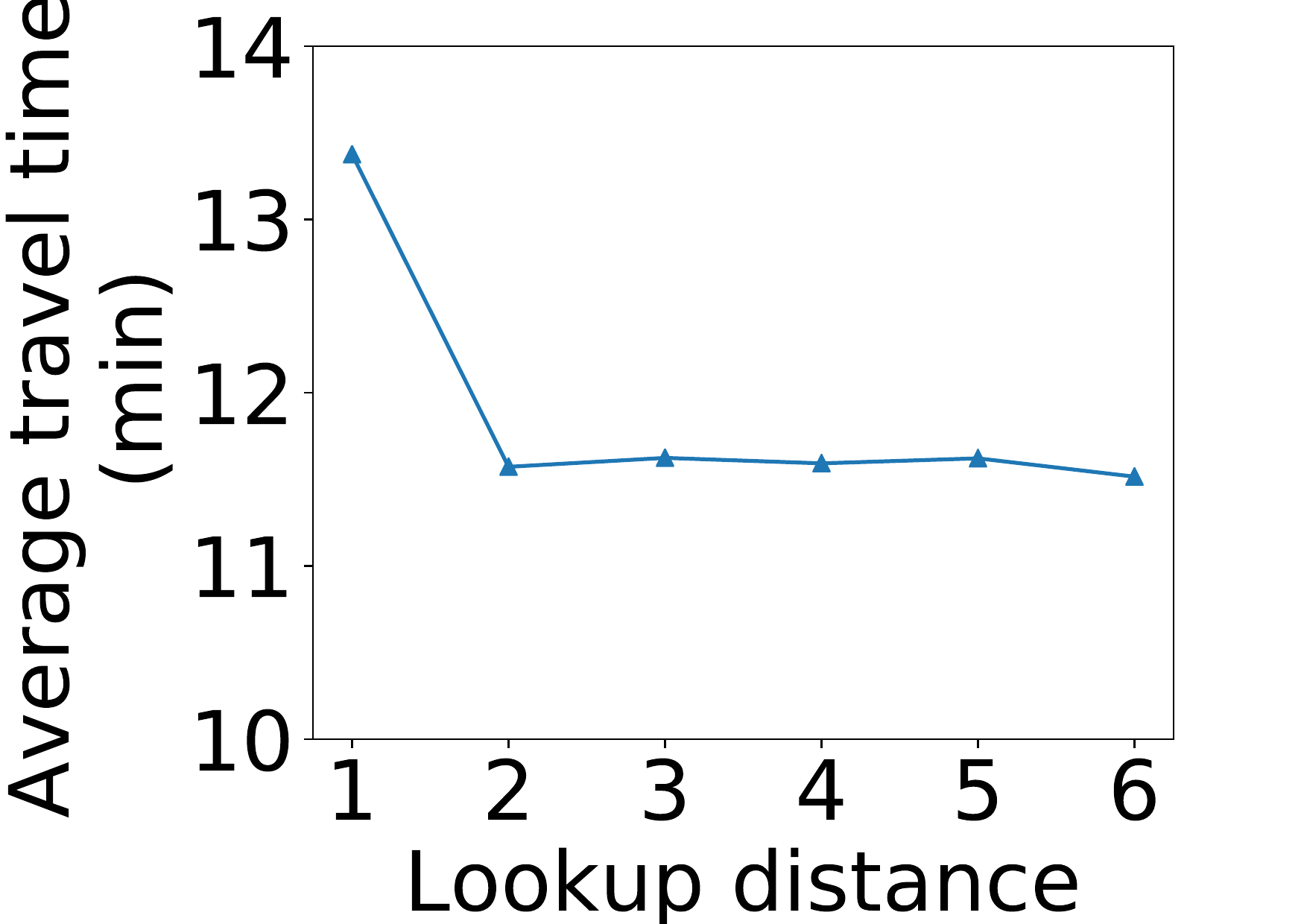}
		\caption{Lookup distance ($l$)}
		\label{length}
	\end{subfigure}
	\caption{Effects of four parameters of CLLA}
	\label{CLLA params}
	\vspace{-2mm}
\end{figure}

\begin{figure}
	\begin{subfigure}[b]{.24\textwidth}
		\centering
		\includegraphics[width=.88\linewidth, height=2.5cm]{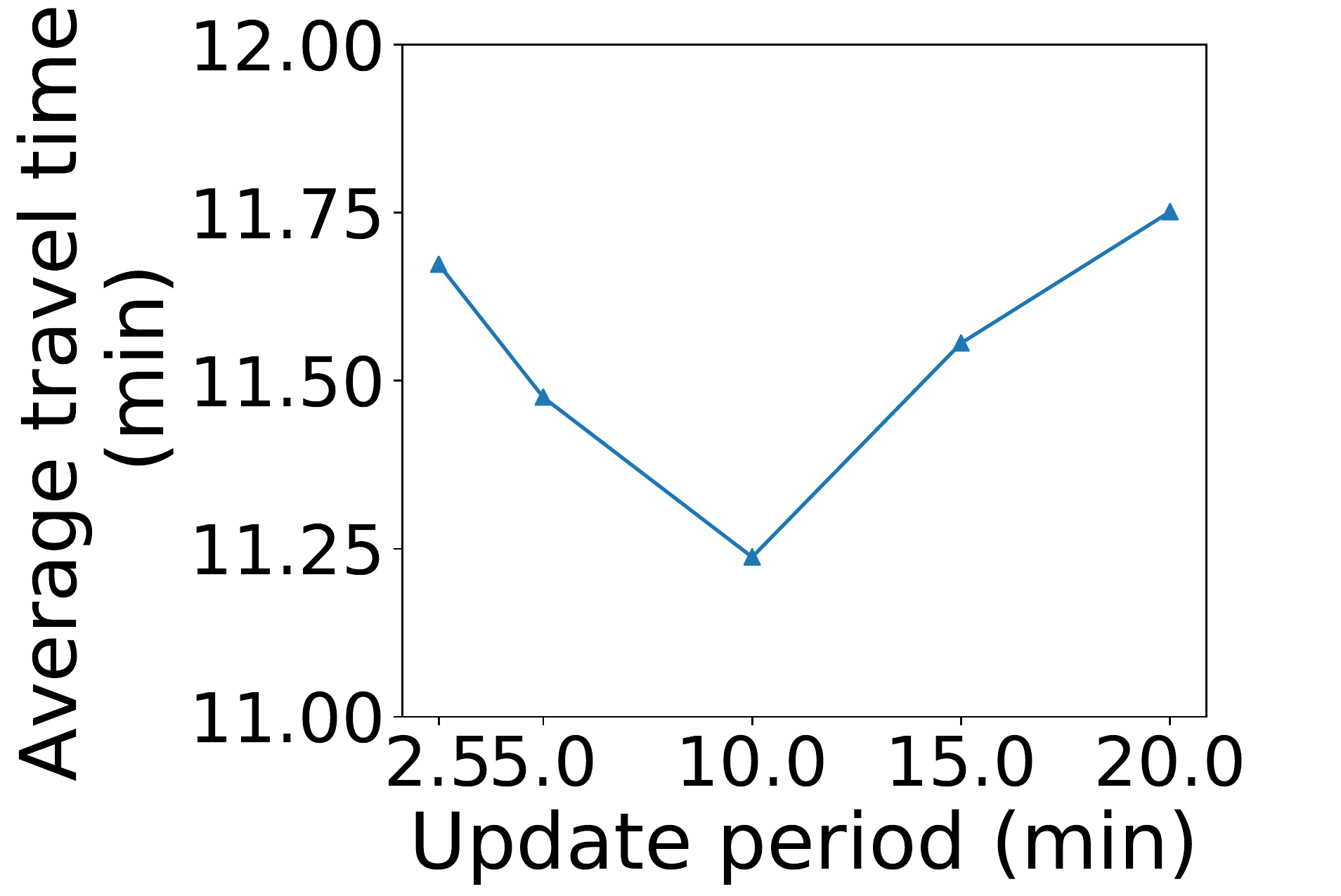}
		\caption{DLA update period}
		\label{DLA update}
	\end{subfigure}%
	\begin{subfigure}[b]{.26\textwidth}
		\centering
		\includegraphics[width=1.04\linewidth, height=2.5cm]{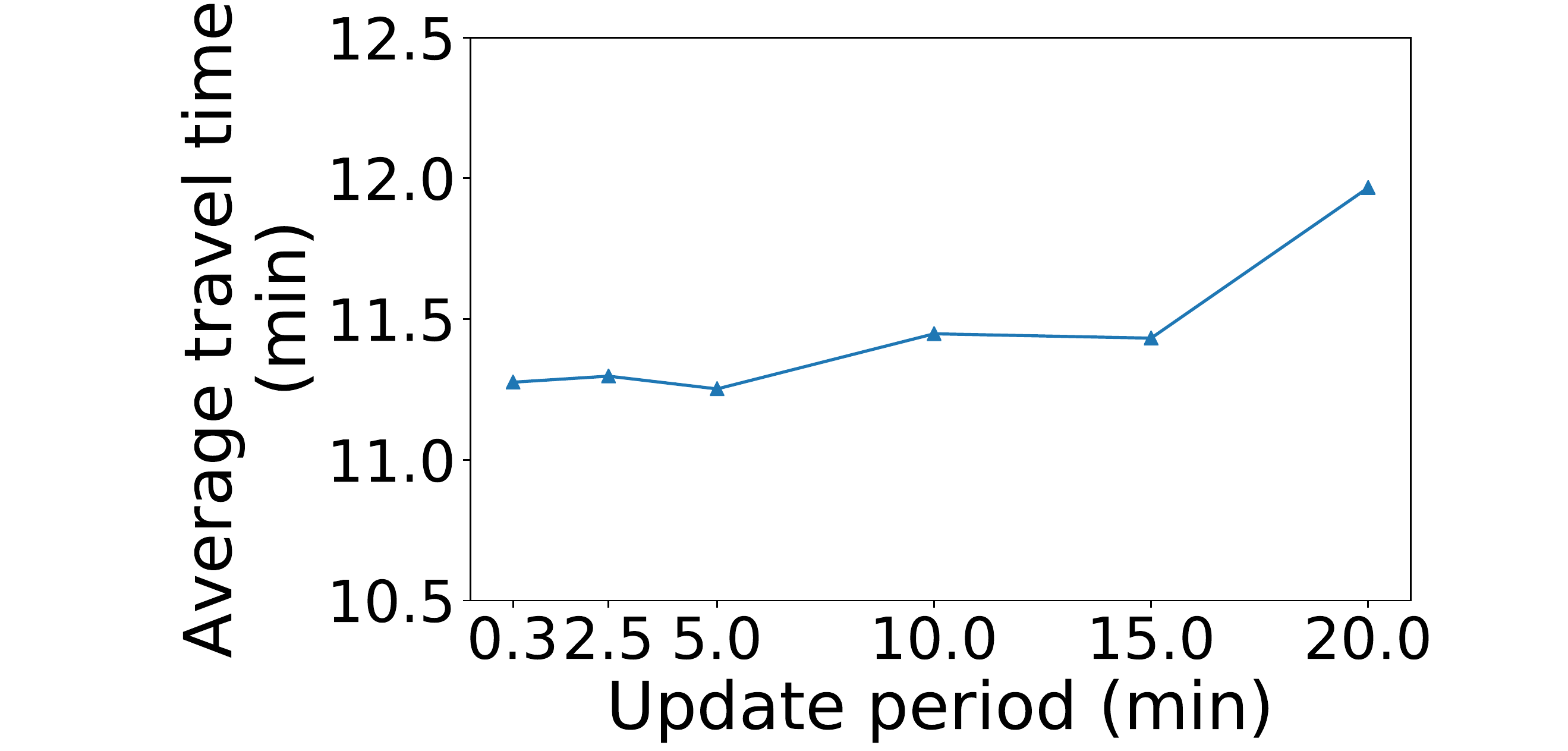}
		\caption{CLLA update period}
		\label{CLLA update}
	\end{subfigure}
	\caption{Effects of the update period of CLLA and DLA}
	\label{update freq: DLA, CCLA}
\end{figure}

\section{Conclusion}

We have shown that effective traffic optimization can be achieved with dynamic lane-direction configurations. Our proposed hierarchical multi-agent solution, CLLA, can help reduce travel time by combining machine learning and the global coordination of lane-direction changes. The proposed solution adapts to significant changes of traffic demand in a timely manner, making it a viable choice for realizing the potential of connected autonomous vehicles in traffic optimization. Compared to state-of-the-art solutions based on lane-direction configuration, CLLA runs more efficiently, and is scalable to large networks.

There are many directions one can investigate further. An interesting extension would be to incorporate dynamic traffic signals into the optimization process. Currently we assume that the connected autonomous vehicles follow the pre-determined path during their trip. An exciting direction for further research is to dynamically change vehicle routes in addition to the lane-direction changes. The dynamic change of speed limit of roads can also be included in an extension to CLLA.

\bibliographystyle{IEEEtran}
\bibliography{icde}

\appendix[Degree of a vertex in PDG]

\begin{lemma} \label{deg lemma}
   The maximum node density $\Delta(PDG)$ is independent of the underline road network $G(V,E)$ and number of vehicle paths. 
\end{lemma}
\begin{proof} 
Let $v\in V^{PDG}$, maximum vertex degree $\Delta(PDG)$ can be found as follows. 

A degree of a vertex $v_{G}\in G(V,E)$ (road network) depends on the number of roads connected to an intersection. A special property of a road graph is that the degree of a node does not increase with the network size. Let there be $n$ number of roads on average, connected to one intersection in $G$. Then $deg(v_{G}) = n$.

Now let us take $v \in V^{PDG}$ is also $v \in E$, where $v$ is a road in the road network. Starting from $v$, within $l$ lookup distance, there can be maximum of $n^l$ roads. Since $n$ does not increase with the network size $n^l$ also does not increase with the network size.

Assuming the worse case, there can be paths from $v$ to each of these $n^l$ roads. Let $R$ be the set of roads in $n^l$. According to the definition of $PDG$, if there is a path between $v$ and $r \in R$, $\exists$ $e_{v,r} \in V^{PDG} \forall r \in R$. This means there are $n^l$ number of edges from $v$. Therefore $deg(v) = n^l$. Then the maximum node density $\Delta(PDG) = n^l$

Note that $\Delta(PDG)$ is independent of the size of $G(V,E)$ and number of paths. 
\end{proof}

\vspace{12pt}

\end{document}